\documentclass[journal,twoside,web]{ieeecolor}
\usepackage{generic}
\usepackage{cite}
\usepackage{textcomp}

\usepackage{subcaption} 
\usepackage{booktabs} 
\usepackage{multirow}
\usepackage[T1]{fontenc}
\usepackage[utf8]{inputenc}
\usepackage{pgfplots}
\usepackage{grffile}
\pgfplotsset{compat=newest}
\usetikzlibrary{plotmarks}
\usetikzlibrary{arrows.meta}
\usepgfplotslibrary{patchplots}
\newcommand{\redsquare}{\tikz\fill[red!60!white] (0,0) rectangle (2mm,2mm);}
\newcommand{\bluesquare}{\tikz\fill[blue!60!white] (0,0) rectangle (2mm,2mm);}
\newcommand{\greensquare}{\tikz\fill[green!70!black] (0,0) rectangle (2mm,2mm);}
\newcommand{\pinksquare}{\tikz\fill[pink!80!white] (0,0) rectangle (2mm,2mm);}
\usepackage{multirow}
\usepackage{booktabs}
\usepackage{tablefootnote}   

\usepackage{amssymb,latexsym,amsfonts,amsmath,bbm}
\usepackage{mathrsfs}
\usepackage{graphicx}
\usepackage{paralist}
\usepackage{dsfont}
\usepackage{eso-pic}
\usepackage{epsfig}
\usepackage{mathtools}
\usepackage{epstopdf}
\usepackage{lipsum}
\usepackage{cite}

\usepackage[nocomma]{optidef}

\usepackage{tikz}
\usetikzlibrary{calc,shapes,arrows}

\usepackage{algorithm}
\usepackage{algorithmic}

\newcounter{theorem}
\newcounter{definition}
\newcounter{lemma}
\newcounter{claim}
\newcounter{problem}
\newcounter{proposition}
\newcounter{corollary}
\newcounter{construction}
\newcounter{example}
\newcounter{xca}
\newcounter{comments}
\newcounter{remark}
\newcounter{assumption}

\newtheorem{theorem}[theorem]{Theorem}
\newtheorem{lemma}[lemma]{Lemma}

\newtheorem{definition}[definition]{Definition}

\newtheorem{remark}[remark]{Remark}
\newtheorem{assumption}[assumption]{Assumption}

\numberwithin{equation}{section}

\DeclareFontFamily{U}{stix2bb}{}
\DeclareFontShape{U}{stix2bb}{m}{n} {<-> stix2-mathbb}{}

\usepackage[many]{tcolorbox}
\usetikzlibrary{calc}
\tcbuselibrary{skins}

\newtcolorbox{resp}[1][]{%
	enhanced jigsaw,%
	colback=gray!5!white,%
	colframe=gray!80!black,%
	size=small,%
	boxrule=1pt,%
	halign title=flush center,%
	coltitle=black,%
	breakable,%
	drop shadow=black!50!white,%
	attach boxed title to top left={xshift=1cm,yshift=-\tcboxedtitleheight/2,yshifttext=-\tcboxedtitleheight/2},%
	minipage boxed title=3cm,%
	boxed title style={%
		colback=white,%
		size=fbox,%
		boxrule=1pt,%
		boxsep=2pt,%
		underlay={%
			\coordinate (dotA) at ($(interior.west) + (-0.5pt,0)$);
			\coordinate (dotB) at ($(interior.east) + (0.5pt,0)$);
			\begin{scope}[gray!80!black]
				\fill (dotA) circle (2pt);
				\fill (dotB) circle (2pt);
			\end{scope}
		}%
	},%
	#1%
}

\DeclareRobustCommand\sampleline[1]{%
	\tikz\draw[#1] (0,0) (0,\the\dimexpr\fontdimen22\textfont2\relax)
	-- (2em,\the\dimexpr\fontdimen22\textfont2\relax);%
}

\usepackage{xspace}
\usepackage{subcaption} 

\newcommand{\R}{{\mathbb{R}}}

\def\BibTeX{{\rm B\kern-.05em{\sc i\kern-.025em b}\kern-.08em
		T\kern-.1667em\lower.7ex\hbox{E}\kern-.125emX}}
\definecolor{blue(ryb)}{rgb}{0.01, 0.28, 1.0}
\definecolor{fashionfuchsia}{rgb}{0.96, 0.0, 0.63}
\makeatletter
\let\NAT@parse\undefined
\makeatother

\definecolor{cerise}{rgb}{0.87, 0.19, 0.39}

\usepackage[colorlinks=true,
citecolor=cerise,
linkcolor=cerise,
urlcolor=cerise,    
final             
]{hyperref}

\renewcommand{\theequation}{\arabic{equation}}
\counterwithout*{equation}{section}

\makeatletter
\def\@opargbegintheorem#1#2#3{\textit{#1\ #2} \textit{(#3):}}
\makeatother

\usepackage{microtype}

\begin{document}
	
\title{From Formal Methods to Data-Driven Safety Certificates of Unknown Large-Scale Networks}
\author{ \IEEEmembership{}	
	\thanks{}
}

\author{Omid Akbarzadeh, \IEEEmembership{Student Member,~IEEE}, Behrad Samari, \IEEEmembership{Student Member,~IEEE}, \\Amy Nejati, \IEEEmembership{Senior Member,~IEEE}, and Abolfazl Lavaei, \IEEEmembership{Senior Member,~IEEE}
	\thanks{All the authors are with the School of Computing, Newcastle University, United Kingdom. Emails:
 \texttt{o.akbarzadeh2@newcastle.ac.uk},
  \texttt{b.samari2@newcastle.ac.uk},
   \texttt{amy.nejati@newcastle.ac.uk},
 \texttt{abolfazl.lavaei@newcastle.ac.uk}. }
}

\maketitle
\begin{abstract}
In this work, we propose a data-driven scheme within a compositional framework with noisy data  to design robust safety controllers in a fully decentralized fashion for large-scale interconnected networks with unknown mathematical dynamics. Despite the network’s high dimensionality and the inherent complexity of its unknown model, which make it intractable, our approach effectively addresses these challenges by (i) treating the network as a composition of smaller subsystems, and (ii) collecting noisy data from each subsystem's trajectory to design a control sub-barrier certificate (CSBC) and its corresponding local controller. To achieve this, our proposed scheme only requires a noise-corrupted \emph{single input-state trajectory} from each unknown subsystem up to a specified time horizon, satisfying a certain rank condition. Subsequently, under a small-gain compositional reasoning, we compose those CSBC, derived from noisy data, and formulate a control barrier certificate (CBC) for the unknown network, ensuring its safety over an infinite time horizon, while providing \emph{correctness guarantees}. We offer a \emph{data-dependent} sum-of-squares (SOS) optimization program for computing CSBC alongside local controllers of subsystems. We illustrate that while the computational complexity of designing a CBC and its safety controller grows \emph{polynomially} with network dimension using SOS optimization, our compositional data-driven approach significantly reduces it to a \emph{linear scale} concerning the \emph{number of subsystems}. We demonstrate the capability of our data-driven approach on multiple physical networks involving unknown models and a range of interconnection topologies.
\end{abstract}

\begin{IEEEkeywords}
Data-driven safety certificates, unknown large-scale networks, compositional techniques, formal methods
\end{IEEEkeywords}

\section{Introduction}\label{sec:intro}

In recent years, large-scale interconnected networks, such as healthcare systems, power grids, and transportation networks, have garnered significant attention owing to their practical applications in real-world scenarios. Formal verification and controller design for such complex networks, aimed at enforcing high-level logic properties including those expressed as linear temporal logic (LTL) formulae~\cite{baier2008principles}, pose inherent challenges. These difficulties primarily stem from (i) the presence of uncountable state and input sets, (ii) the large dimensionality of the underlying network, (iii) the complexity of logical requirements, and (iv) the \emph{absence} of mathematical models in numerous real-life scenarios. To address the underlying difficulties, existing literature predominantly relies on utilizing \emph{finite abstractions} to approximate original models with simpler ones featuring discrete state sets~\cite{tabuada2009verification}. However, abstraction-based techniques, centered on discretizing state and input sets, may not be applicable to practical applications due to the \emph{exponential} state-explosion problem (see \emph{e.g.,}~\cite{APLS08,Girard-Approximation,julius2009approximations,zamani2014symbolic}).

An alternative and promising solution for formal analysis of dynamical systems, aimed at circumventing the state explosion problem, involves the utilization of control barrier certificates, as a \emph{discretization-free} approach, initially proposed in~\cite{Pranja}. In particular, barrier certificates, akin to Lyapunov-like functions, are designed to satisfy specific conditions on the barrier function itself and its evolution over the flow of the system. By establishing an initial level set of barrier certificates from a set of initial states, an unsafe region can be separated from system trajectories. Consequently, such barrier certificates offer a formal (probabilistic) guarantee over the safety of the system (see \emph{e.g.,}~\cite{ames2019control,wieland2007constructive,xiao2023safe, NEURIPS2020_barrier, NEURIPS2021_Barrier,lavaei2022compositional,wooding2024protect}).

{\bf Primary Impediments.}  While barrier certificates offer formal analysis capabilities over dynamical systems, they also introduce two primary challenges. First, existing methods for searching for barrier certificates, like SOS optimization, do not scale well with system dimensionality, often leading to computational intractability, particularly when dealing with large-scale networks. To address this issue, a promising solution involves the development of a compositional approach, which treats a large-scale network as an interconnection of lower-dimensional subsystems and aims to design control sub-barrier certificates at the level of subsystems~\cite{lavaei2024scalable,zaker2024compositional,ref4}.  

It is worth noting that while compositional techniques can be used for both abstraction-based and discretization-free approaches, their use differs fundamentally. More precisely, since in abstraction-based methods constructing a symbolic model for the entire network is infeasible due to high dimensionality, symbolic models of subsystems are first constructed, by establishing similarity relations between actual and abstract subsystems. Then, compositional frameworks come into play, upon which it is possible to compositionally construct the symbolic model of the network together with its associated similarity relation (see \textit{e.g.,} \cite{swikir2019compositional,nejati2020compositional,lavaei2019automated}). In discretization-free approaches, however, similar to our work, the goal is to adopt a compositional framework that allows us to formally construct a safety certificate for the entire network, provided that each subsystem's conditions and certain compositional conditions are satisfied (see \textit{e.g.,}~\cite{nejati2022dissipativity,lavaei2024scalable}).

As for the second challenge, designing barrier certificates necessitates \emph{precise knowledge} of the system's mathematical dynamics, rendering model-based approaches impractical. Hence, data-driven techniques have gained popularity as they enable the utilization of system measurements for analysis, circumventing the need for explicit models. These methods can be categorized into \emph{indirect} and \emph{direct} approaches, each with distinct methodologies and application areas~\cite{Hou2013model,dorfler2022bridging,nejati2023formal,ROTULO2022110519,nejati2023data,lavaei2023symbolic}.

Specifically, \emph{indirect} data-driven methods involve system identification, leveraging a wide array of powerful tools from model-based control techniques after the identification phase. However, it faces challenges with computational complexity in \emph{two} key phases: first, during model identification, and later in solving the model-based problem.  
In contrast, \emph{direct} data-driven approaches skip the system identification step, applying system measurements directly to analyze unknown systems~\cite{dorfler2022bridging}. Despite the emergence of data-driven approaches for safety certifications of systems with unknown dynamics (\emph{e.g.,}~\cite{bisoffi2020controller,Bartocci-Data-Driven,nejati2022data,akbarzadeh2024learning,samari2024single}), these methods are primarily limited to small-scale systems and cannot handle large-scale networks with potentially over $1000$ dimensions due to the significant sample complexity problem.

We note that under the category of direct data-driven techniques, \emph{scenario-based} approaches \cite{calafiore2006scenario} can address the safety problem by gathering data from the unknown system, typically through intermediate steps involving chance constraints \cite{esfahani2014performance}. While these approaches show significant promise in providing formal guarantees for unknown systems, they rely on the assumption that the data is \emph{independent and identically distributed (i.i.d.)}. Consequently, only one input-output data pair can be extracted from each trajectory \cite{calafiore2006scenario}, necessitating the collection of multiple independent trajectories—typically millions in practical cases—to achieve a specified confidence level, as determined by a known closed-form relationship~\cite{nejati2023formal}.
	
In contrast, our data-driven approach does not necessitate having i.i.d. data, implying that one trajectory per subsystem would suffice to conduct the analysis if a rank condition is satisfied (cf. Remark~\ref{Rank-condition}). Moreover, the guarantee offered within our framework is not probabilistic, unlike the scenario approach with random sampling, which involves some degree of risk and confidence.
Despite these advantages, we should acknowledge that the scenario-based methods do not require the dictionary of monomials $\mathcal{M}_i(x_i)$, which is required in our framework. Additionally, these approaches accommodate general nonlinear systems, whereas our analysis is tailored to nonlinear systems with polynomial dynamics.

It is important to note that data-driven CBC methods developed for monolithic systems using single trajectories are not directly applicable to interconnected networks due to the computational complexity of the SOS program. We refer to related work \cite{nejati2022data}, where the proposed approach was limited to handling $2$-dimensional systems in its case study section, in contrast to our simulation results, which demonstrate scalability to systems with $2000$ dimensions.

{\bf Innovative Findings.} Our work introduces a \emph{compositional data-driven}  approach with noisy data to design fully-decentralized robust controllers that enforce safety properties over large-scale interconnected networks with unknown mathematical models.  Our framework employs control sub-barrier certificates, derived from noisy data, to synthesize a local controller for each subsystem. This is achieved based solely on a  noise-corrupted input-state trajectory collected from subsystems, subject to a specific rank condition associated with (generalization of) the persistency of excitation \cite{persistency}. We then offer a small-gain compositional condition to integrate data-driven CSBC and establish a control barrier certificate for the unknown network, while guaranteeing its safety over an infinite time horizon. We showcase the effectiveness of our data-driven results across a variety of physical benchmarks with unknown models and different interconnection topologies.

\section{System Description}{\label{sec:Continuous}}

\subsection{Notation}
Sets of real, non-negative and positive real numbers are represented by $\mathbb{R}$, $\mathbb{R}^{+}_0$, and $\mathbb{R}^+$, respectively. The set of non-negative integers is denoted as $\mathbb{N} := \{0,1,2,\ldots\}$, while the set of positive integers is signified by $\mathbb{N}^+ = \{1,2,\ldots\}$. For $N$ vectors $x_i$ in $\mathbb{R}^{n_i}$, the notation $x=[x_1;\ldots;x_N]$ represents a column vector formed by these vectors, with a total dimension of $\sum_i n_i$. We use $\begin{bmatrix} x_1 & \ldots & x_N\end{bmatrix}$ and $\begin{bmatrix} A_1 & \ldots & A_N\end{bmatrix}$ to represent the horizontal concatenation of vectors $x_i \in \mathbb{R}^n$ and matrices $A_i \in \mathbb{R}^{n \times m}$ to form $n \times N$ and $n \times m N$ matrices, respectively. The notation $\{a_{ij}\}$ represents a matrix formed by placing the elements $a_{ij}$ in the $i$-th row and $j$-th column. We use $\Vert \cdot \Vert$ to denote the Euclidean norm for a vector $x \in \mathbb{R}^n$ and the induced 2-norm for a matrix $A \in \mathbb{R}^{n \times m}$. For sets $X_i$, where $i \in \{1,\ldots,N\}$, their Cartesian product is denoted as $\prod_{i=1}^{N}X_i$. An identity matrix in $\mathbb{R}^{n \times n}$ is indicated by $\mathds{I}_n$, while $\mathbf{1}_n$ symbolizes a column vector of dimension $\mathbb{R}^{n \times 1}$, with all entries equal to one. A zero matrix of dimension $n\times m$ is denoted by  $\mathbf{0}_{n\times m}$. An empty set is represented by $ \emptyset$.  A (block) diagonal matrix in $\R^{N\times{N}}$ with  diagonal matrix entries $(A_1,\ldots,A_N)$ and scalar entries $(a_1,\ldots,a_N)$ is denoted by $\mathsf{blkdiag}(A_1,\ldots,A_N)$ and $\mathsf{diag}(a_1,\ldots,a_N)$, respectively. A \emph{symmetric} and positive-definite matrix $P \in\mathbb{R}^{n\times n}$ is denoted by $P \succ 0$, while $P \succeq 0$ denotes that $P$ is a \emph{symmetric} positive semi-definite matrix. The transpose of a matrix $P$ is represented by $P^\top$\!, while its left pseudoinverse is represented by $P^{\dagger}$.

\subsection{Individual Subsystems}\label{systems1}

In this work, our focus lies on continuous-time nonlinear polynomial systems, treating them as individual subsystems, as delineated in the subsequent definition.
\begin{definition}
A continuous-time nonlinear polynomial system (ct-NPS) is described by 
\begin{subequations} 
\begin{align}\label{sys2}
\Theta_i\!: \dot x_i=A_i\mathcal M_i(x_i) + B_i\nu_i + D_i w_i,
\end{align}
where $\mathcal M_i(x_i) \in \mathbb R^{N_i}$, with $\mathcal{M}_i(\mathbf{0}_{n_i})=\mathbf{0}_{N_i}$, is a vector of monomials in states $x_i\in X_i$, $A_i \in \mathbb R^{n_i\times N_i}, B_i \in \mathbb R^{n_i\times m_i},$ $D_i \in \R^{n_i \times \sigma_i}$ with $\sigma_i = \sum_{j = 1, j \ne i}^{\mathcal{Q}} n_j$, and $\mathcal{Q} $ is the number of subsystems interconnected with $\Theta_i$. The vector-valued function $\mathcal M_i(x_i)$ satisfies
\begin{equation}\label{Transformation}
\mathcal{M}_i(x_i)= \Upsilon_i(x_i) x_i,
\end{equation}
with a state-dependent transformation matrix $\Upsilon_i(x_i)\in  \mathbb R^{N_i \times n_i}$. Furthermore, $\nu_i \in U_i$ and $w_i \in W_i$ are \emph{external} and \emph{internal} inputs of ct-NPS, with $X_i\subseteq \mathbb R^{n_i}$, $U_i \subseteq \mathbb R^{m_i}$, and $W_i \subseteq \mathbb{R}^{\sigma_i}$ being state, external and internal input sets, respectively. We utilize the tuple $\Theta_i = (A_i, B_i, D_i, X_i,U_i, W_i)$ to represent the ct-NPS.

\noindent Under the transformation~\eqref{Transformation}, the ct-NPS in~\eqref{sys2} can be equivalently  written as
\begin{align}\label{sys3}
	\Theta_i\!: \dot x_i=A_i \Upsilon_i(x_i)x_i + B_i\nu_i + D_i w_i.
\end{align}
\end{subequations} 
\end{definition}
\begin{remark}
Transformation~$\Upsilon_i(x_i)$ in~\eqref{Transformation} ensures that all expressions are ultimately expressed in terms of $x_i$ rather than $\mathcal{M}_i(x_i)$. This consistency is essential since our CSBC is later defined as $\mathds{B}_i(x_i) = x_i^\top P_i x_i$ (cf. Theorem~\ref{Thm:main}), which depends solely on $x_i$. Expressing everything in terms of $x_i$ simplifies the formulation and enables us to propose a set of more scalable conditions in Theorem~\ref{Thm:main}. Moreover, we note that, without loss of generality, for any $\mathcal{M}_i(x_i)$ satisfying $\mathcal{M}_i(\mathbf{0}_{n_i}) = \mathbf{0}_{N_i}$, there always exists a transformation matrix $\Upsilon_i(x_i)$ that satisfies~\eqref{Transformation}.
\end{remark}
In system \eqref{sys2}, we designate matrices $A_i$ and $B_i$ as unknown, defining this setup as the \emph{unknown model}. On the contrary, we assume knowledge of matrix $D_i$, as it represents interconnection weights among subsystems, often predetermined in interconnected networks.

\begin{remark}\label{rem:dict}
		While the exact $\mathcal{M}_i(x_i)$ is unknown, we assume that its dictionary (\emph{i.e.,} a library or family of functions) can be constructed to capture the actual dynamics by being \emph{sufficiently extensive}. Specifically, such a dictionary is structured to be comprehensive, capturing all possible terms in the actual system's dynamics, even if it includes some terms that are ultimately unnecessary. To achieve this, an upper bound on the maximum degree of $\mathcal{M}_i(x_i)$ can be obtained based on the physical insights into the system, allowing $\mathcal{M}_i(x_i)$ to be constructed to include all possible combinations of states up to that bound (cf. all benchmark case studies).
\end{remark}

It is worth highlighting that many natural systems evolve in continuous time with nonlinear dynamics \cite{strogatz2018nonlinear}, which is the focus of this work. We restrict subsystems to nonlinear \emph{polynomial} systems to ensure computational tractability via SOS programming. While our framework could, in principle, extend to other nonlinear systems (\emph{e.g.,} with sinusoidal terms), key conditions such as~\eqref{con11} in Theorem~\ref{Thm:main} may no longer be solvable using SOS tools. The focus on polynomial dynamics is also further motivated by their prevalence in modeling engineering systems, including fluid dynamics and robotics.

\subsection{Interconnected Network}\label{In-Net}
Since our primary objective is to analyze the safety of interconnected networks, composed of individual subsystems represented by \eqref{sys2}, this subsection outlines the network configuration and details the interconnection among subsystems.
Here, we present a formal definition of an interconnected network comprising $\mathcal{Q} \in \mathbb{N}^{+}$ individual subsystems, each designated as $\Theta_i$, with partitioned \emph{internal} inputs and their corresponding matrices as
\begin{subequations} 
\begin{align}\label{in-out}
w_i &=\left[w_{i 1} ; \ldots ; w_{i(i-1)} ; w_{i(i+1)} ; \ldots ; w_{i \mathcal{Q}}\right]\!\!,\\ \label{D-partition}
D_i&=\left[
		D_{i 1}\,\,  \ldots\,\,  D_{i(i-1)}\,\,  D_{i(i+1)} \,\, \ldots \,\, D_{i N} \right]\!\!,
\end{align}
\end{subequations} 
where $D_{i j} \in \mathbb{R}^{n_i \times n_j}$. We assume that the dimension of $w_{ij}$ and $x_{j}$ are matched, which is a well-defined assumption within the framework of small-gain reasoning. Moreover, in cases where there is no connection from subsystem $\Theta_i$ to $\Theta_j$, we presume the corresponding internal input to be zero, \emph{i.e.}, $w_{i j} \equiv 0$, otherwise, $w_{i j}=x_j$.

We now proceed to define the interconnected network as follows.
\begin{definition}\label{network}
Consider $\mathcal{Q} \in \mathbb{N}^{+}$ subsystems $\Theta_i\!=\!(A_i, B_i, D_i, X_i,U_i, W_i)$, $i \in\{1, \ldots, \mathcal{Q}\}$, with the internal input configuration as in~\eqref{in-out}. The interconnection of $\Theta_i$, $i \in\{1, \ldots, \mathcal{Q}\}$, forms the interconnected network $\Theta=\left(A(x), B, X, U \right)$, denoted by $\mathcal{I}\left(\Theta_1, \ldots, \Theta_\mathcal{Q}\right)$, such that $X:=\prod_{i=1}^\mathcal{Q} X_i, U:=\prod_{i=1}^\mathcal{Q} U_i$, subjected to the following interconnection constraint:
\begin{equation}\label{Internal}
\forall i, j \in\{1, \ldots, \mathcal{Q}\}, i \neq j: \quad \omega_{i j}=x_{j}.
\end{equation}
Such an interconnected network can be described by  
\begin{align}\label{sysN1}
\Theta\!: \dot x=A(x)x + B\nu,
\end{align}
where $A(x) \in \mathbb{R}^{n \times n}$ is a block matrix, with $A(x) = \{\mathsf{a}_{ij}\}$, $i,j \in\{1, \ldots, \mathcal{Q}\}$, and $n = \sum_{i=1}^\mathcal{Q} n_i$, consisting of diagonal blocks $(A_1\Upsilon_1(x_1), \ldots, A_\mathcal{Q}\Upsilon_\mathcal{Q}(x_\mathcal{Q}))$ and off-diagonal blocks ${D}_{ij} $. The interconnection topology determines the values of $D_i$~in~\eqref{D-partition}. Moreover, $B = \mathsf{blkdiag}(B_1, \ldots, B_\mathcal{Q}) \in \mathbb{R}^{n \times m}$, $\nu \in \mathbb{R}^{m}$, with $m = \sum_{i=1}^\mathcal{Q} m_i$, and $x = [x_1; \dots; x_\mathcal{Q}] \in \mathbb{R}^{n}$\!.
\end{definition}
We employ $x_{x_{0,\nu}}(t)$ to denote the solution process of the interconnected network at time $t \in \mathbb{R}^{+}_{0}$ under input trajectories $\nu(\cdot)$, originating from an initial state $x_0 \in X$.

\subsection{Control (Sub-)Barrier Certificates}\label{Barrier}
We present in this subsection the concept of \emph{control sub-barrier and barrier certificates} for both individual subsystems and interconnected networks, \emph{with and without} internal signals, respectively.
In the subsequent definition, we initially present the concept of control sub-barrier certificates for ct-NPS with internal inputs.
\begin{definition}\label{csbc}
Consider a ct-NPS
$\Theta_i\!=\!(\!A_i, B_i, D_i, X_i,U_i, W_i)$, with $X_{0_i}, X_{u_i} \subseteq X_i$ being its initial and unsafe sets, respectively. Assuming the existence of constants $\phi_i,\gamma_i,\beta_i,\varepsilon_i\in \mathbb{R}^{+}$ and $\rho_i \in \mathbb{R}^{+}_{0}$, a function $\mathds B_i:X_i\to\R_0^+$ is called a control sub-barrier certificate (CSBC) for $\Theta_i$ if
\begin{subequations}\label{subsys}
\begin{align}\label{subsys1}
 &\mathds{B}_i(x_i) \geq \phi_i\|x_i\|^2, \:\:\:\:\:\quad\quad\quad\quad\quad\quad\!\!\!\!\!\forall x_i \in X_i,\\\label{subsys11}
&\mathds B_i(x_i) \leq \gamma_i,\:\:\:\:\quad\quad\quad\quad\quad\quad\quad\quad\!\forall x_i \in X_{0_i},\\\label{subsys111}
&\mathds B_i(x_i) \geq \beta_i, \:\:\:\:\quad\quad\quad\quad\quad\quad\quad\quad\!\forall x_i \in X_{u_i}, 
\end{align}  
and $\forall x_i\in X_i$, $\exists \nu_i\in U_i$, such that $\forall w_i\in W_i$, one has 
\begin{align}\label{eq:martingale2}
	&\mathcal{L}\mathds B_i(x_i)\mathcal ~\!\leq - \varepsilon_i \mathds B_i(x_i) + \rho_i\Vert w_i\Vert^2,
\end{align}
\end{subequations}
where $\varepsilon_i$ and $\rho_i$ denote the decay rate and the interaction gain, respectively, and   $\mathcal{L} \mathds B_i$ is the Lie derivative of $\mathds B_i$
with respect to ct-NPS in~\eqref{sys3}, which is defined as~\cite{Willmore_1960} 
\begin{align}\label{Lie derivative}
	\mathcal{L}\mathds B_i( x_i)=\partial_{x_i} \mathds B_i(x_i)\big(A_i\Upsilon_i(x_i)x_i+ B_i\nu_i + D_i w_i\big)\!,
\end{align}
with $\partial_{x_i} \mathds B_i(x_i) = \frac{\partial\mathds B_i(x_i)}{\partial_{x_i}}$.
\end{definition}

We now provide the following definition to describe control barrier certificates for \emph{interconnected networks} without internal signals.
\begin{definition} \label{cbc}
	Consider a network $\Theta = \mathcal{I}\left(\Theta_1, \ldots, \Theta_\mathcal{Q}\right)$, as in Definition~\ref{network}, composed of $\mathcal{Q}$ subsystems $\Theta_i$. Assuming  the existence of  constants $\gamma,\beta,\varepsilon\in\mathbb{R}^{+}$, with $\beta > \gamma$, a function $\mathds B:X \rightarrow \mathbb{R}_{0}^{+}$ is referred to as a control barrier certificate (CBC) for $\Theta$, if 
 \begin{subequations}
	\begin{align}\label{CBC1}
		&\mathds B(x) \leq \gamma,\quad\quad\quad\quad\quad\quad\!\!\forall x \in X_{0},\\\label{CBC2}
		&\mathds B(x) \geq \beta, \quad\quad\quad\quad\quad\quad\!\!\forall x \in X_{u},
	\end{align}  
	and $\forall x\in X$, $\exists \nu\in U$,  such that
	\begin{align}\label{cbceq}
		\mathcal{L}\mathds B( x)\mathcal ~\!\leq - \varepsilon \mathds B(x),
	\end{align}
  \end{subequations}
	where $\mathcal{L} \mathds B$ represents the Lie derivative of $\mathds B$
	with respect to interconnected dynamics in \eqref{sysN1}, defined as 
	\begin{align}
	\mathcal{L}\mathds B (x)=\partial_{x} \mathds B (x)(A(x)x + B \nu).
	\end{align}
	Moreover, $X_0, X_{u} \subseteq X$ denote the initial and unsafe sets of the interconnected network, respectively. 
\end{definition}

\begin{remark}
In the CBC framework adopted in this work, the initial set $X_0$ and the unsafe set \( X_u\) are predetermined given the property of interest, and it is required that $X_0 \,\cap\, X_u \;=\;\emptyset$. If this condition is violated, the system is unsafe from the outset, and any safety analysis becomes meaningless. Moreover, for the sake of using the SOS programming, the initial, unsafe, and state sets, should be all semialgebraic,  \emph{e.g.,} describable by a finite collection of polynomial inequalities.
\end{remark}
We now utilize CBC to establish safety guarantees across the interconnected network over an infinite time horizon, as presented in the following theorem~\cite{Pranja}.
\begin{theorem}\label{Theo1}
Consider a network $\Theta = \mathcal{I}\left(\Theta_1, \ldots, \Theta_\mathcal{Q}\right)$, composed of $\mathcal{Q}$ subsystems $\Theta_i$.  Suppose $\mathds{B}$ is a CBC for $\Theta$ as in Definition~\ref{cbc}. Then, the network is safe under a controller $\nu(\cdot)$ within an infinite time horizon, \emph{i.e.,} $x_{x_{0,\nu}}(t) \cap X_{u}=\emptyset$ for any $x_{0} \in X_0$ and any $t \in \mathbb{R}^{+}_{0}$.
\end{theorem}

\begin{proof} 
	The proof is demonstrated by contradiction. Assume $x_{x_{0,\nu}}(t)$ of $\Theta$ originates from $x_0 \in X_0$ and enters $X_{u}$ under a controller $\nu(\cdot)$. According to conditions~\eqref{CBC1} and~\eqref{CBC2}, one has $\mathds{B}(x(0)) \leq \gamma$ and $\mathds{B}(x(t)) \geq \beta$ for $t \in \mathbb{R}^{+}_{0}$. Since $\mathds{B}(x(\cdot))$ serves as a CBC and in light of condition~\eqref{cbceq}, it follows that $\beta \leq \mathds{B}(x(t)) \leq \mathds{B}(x(0)) \leq \gamma$. This contradicts the essential condition $\beta > \gamma$, thereby concluding the proof by contradiction. 
\end{proof} 

Generally, designing a CBC and its corresponding controller to uphold safety across a large-scale system is computationally very burdensome, even with a known underlying model. This is why our compositional method focuses on subsystem levels, initially constructing CSBC and then, under specific compositional conditions, forming a CBC from CSBC of subsystems. Nevertheless, it is clear that directly fulfilling condition~\eqref{eq:martingale2} for constructing CSBC is not practical, due to the inclusion of unknown matrices $A_i$ and $B_i$ in $\mathcal {L}\mathds {B}_i(x_i)$. With these main challenges at hand, we now formalize the primary problem we aim to tackle in this work.
\begin{resp}
\textbf{Problem Formulation:}\label{problem}
Consider an interconnected network $\Theta = \mathcal{I}\left(\Theta_1, \ldots, \Theta_\mathcal{Q}\right)$, composed of $\mathcal{Q}$ subsystems $\Theta_i$, each with unknown matrices $A_i$ and $B_i$.  By collecting only a single input-state trajectory from each subsystem, construct CSBCs and their decentralized controllers. Subsequently, devise a compositional methodology, grounded in small-gain reasoning, to integrate these CSBCs and design a CBC and its corresponding controller, while ensuring safety properties across the interconnected network.
\end{resp}

\section{Data-Driven Design of CSBC and Local Controllers}\label{DD-CBCs}
To tackle the stated problem, this section describes our data-driven framework, aimed at constructing a CSBC and its associated decentralized controller from data for each unknown subsystem. In our data-driven scheme, we initially define the structure of our CSBC as quadratic, expressed as $\mathds{B}_i(x_i) = x_i^\top P_i x_i$, with $P_i \succ 0$. Subsequently, we gather input-state data from unknown ct-NPS during the time interval $[t_0, t_0 + (\mathcal{T} - 1)\tau]$, with $\mathcal{T} \in \mathbb N^{+}$ representing the number of collected samples, and $\tau \in \mathbb R^{+}$ denoting the sampling interval:
\begin{subequations}\label{New}
\begin{align}
\mathcal U^{0,\mathcal{T}}_i &\!\!=\! [\nu_i(t_0)~~\nu_i(t_0 + \tau)~\dots~\nu_i(t_0 + (\mathcal{T} - 1)\tau)] \label{eq:U0},\\
\mathcal W^{0,\mathcal{T}}_i &\!\!=\! [w_i(t_0)~w_i(t_0 + \tau)~\dots~w_i(t_0 + (\mathcal{T} - 1)\tau)] \label{eq:W0},\\
\mathcal X^{0,\mathcal{T}}_i &\!\!=\! [x_i(t_0)~~x_i(t_0 + \tau)~\dots~x_i(t_0 + (\mathcal{T} - 1)\tau)] \label{eq:X0},\\
\bar{\mathcal X}^{1,\mathcal{T}}_i &\!\!=\! [\dot x_i(t_0)~~\dot x_i(t_0 + \tau)~\dots~\dot x_i(t_0 + (\mathcal{T} - 1)\tau)] \label{eq:barX1}.
\end{align}
\end{subequations}
Since direct measurement of the state derivatives at the sampling instants, defined in~\eqref{eq:barX1}, is challenging, we consider that these measurements are \emph{corrupted by noise}, denoted as $\Gamma_i:= [\,\varkappa_i(t_0),\,\varkappa_i(t_0 + \tau),\,\dots,\,\varkappa_i(t_0 + (\mathcal{T}-1)\tau)\,] \in \mathbb{R}^{n_i \times\mathcal{T}}$. Consequently, our measured data becomes $\mathcal{X}_i^{1,\mathcal{T}}:= \bar{\mathcal{X}}_i^{1,\mathcal{T}} + \Gamma_i$, where the noise term $\Gamma_i$ is \emph{unknown} but adheres to the subsequent assumption.
\begin{assumption}
The unknown noise matrix $\Gamma_i$ satisfies
\begin{equation}\label{noise-bound}
\Gamma_i \Gamma_i^{\top} \preceq \Xi_i \Xi_i^{\top}\!,
\end{equation}
for a known matrix $\Xi_i \in \mathbb{R}^{n_i \times \mathcal{T}}$, which essentially means that the energy of the noise remains bounded over the finite time of data collection~\cite{van2020noisy}.
\end{assumption}
In practice, this condition is automatically fulfilled if there exists a constant $\bar{\varkappa_i} \in \mathbb{R}^{+}$ such that $\|\varkappa_i(t)\|^2 \leq \bar{\varkappa_i}$ for all $t \in \mathbb{R}_{ 0}^{+}$. In such a case, we have
\begin{equation}\label{noise-bound-2}
\Xi_i \Xi_i^{\top} = \bar{\varkappa}_i \mathcal{T} \mathds{I}_{n_i}.
\end{equation}

\begin{remark}
		The system's safety necessitates collecting data while the system functions within safe operating regions. This requirement is solely due to physical/practical limitations and is not related to our theoretical framework.
\end{remark}

 We consider the set of trajectories in~\eqref{New} as a \emph{single input-state trajectory}. Given the presence of unknown matrices $A_i$ and $B_i$ in $\mathcal {L}\mathds{B}_i(x_i)$, inspired by~\cite{de2019formulas}, we introduce the following lemma to derive the data-driven representation of $A_i\mathcal{M}_i(x_i) + B_i\nu_i$.
\begin{lemma}\label{Lemma1}
	Given subsystems $\Theta_i$ in~\eqref{sys2}, let $\mathcal {S}_i (x_i)$ be a $(\mathcal{T}\times n_i)$ polynomial matrix such that 
 \begin{equation}\label{martrixS}
  \Upsilon_i(x_i) = \mathcal N^{0,\mathcal{T}}_i \mathcal {S}_i (x_i),
\end{equation}
with $ \Upsilon_i(x_i) \in \mathbb{R}^{N_i \times n_i}$ being the state-dependent transformation matrix, as in~\eqref{Transformation}, and 
	\begin{align}\label{trajN}\nonumber
		\mathcal N^{0,\mathcal{T}}_i & = \Big[\overbrace{\Upsilon_i(x_i(t_0))x_i(t_0)}^{\mathcal M_i(x_i(t_0))}~~\overbrace{\Upsilon_i(x_i(t_0 + \tau))x_i(t_0 + \tau)}^{\mathcal M_i (x_i(t_0 + \tau))}~\dots~~\\& ~~~~\underbrace{\Upsilon_i (x_i(t_0 + (\mathcal{T} - 1)\tau))(x_i(t_0 + (\mathcal{T} - 1)\tau)}_{\mathcal M_i (x_i(t_0 + (\mathcal{T} - 1)\tau))}\Big],
	\end{align}
	being a \emph{full row-rank}, $(N_i\times \mathcal{T})$ matrix, constructed from the vector $\mathcal M_i(x_i)=\Upsilon_i(x_i)x_i$~in~\eqref{Transformation} and samples $\mathcal X^{0,\mathcal{T}}_i$\!\!. By designing a decentralized controller $\nu_i = \mathcal{F}_i(x_i)\,x_i = \mathcal{U}^{0,\mathcal{T}}_i \mathcal{S}_i(x_i)\,x_i$, with $\mathcal{F}_i(x_i) = \mathcal{U}^{0,\mathcal{T}}_i \mathcal{S}_i(x_i)$, the expression $A_i \Upsilon_i(x_i) + B_i\mathcal{F}_i(x_i)$ can be equivalently described by the following data-driven representation:
 \begin{equation*}
    A_i \Upsilon_i(x_i) + B_i \mathcal{F}_i(x_i)= (\mathcal X^{1,\mathcal{T}}_i - D_i\mathcal W^{0,\mathcal{T}}_i - \Gamma_i) \mathcal {S}_i (x_i). 
 \end{equation*}
\end{lemma}
\begin{proof} 
	Using trajectories~\eqref{New}, one can obtain the data-based representation of $\dot x_i=A_i\mathcal M_i(x_i) + B_i\nu_i + D_i w_i$ as the following:
\begin{align*}
	\bar{\mathcal X}^{1,\mathcal{T}}_i &= A_i \mathcal N^{0,\mathcal{T}}_i  + B_i\mathcal U^{0,\mathcal{T}}_i  + D_i\mathcal W^{0,\mathcal{T}}_i = [B_i\quad A_i] \begin{bmatrix}
		\mathcal U^{0,\mathcal{T}}_i \\
		\mathcal N^{0,\mathcal{T}}_i 
	\end{bmatrix}\\ &~~~ + D_i\mathcal W^{0,\mathcal{T}}_i\!\!.
\end{align*}
Accordingly, as $\mathcal{X}_i^{1,\mathcal{T}}= \bar{\mathcal{X}}_{i}^{1,\mathcal{T}} + \Gamma_i$, one has
\begin{align} \label{data3}
[B_i\quad A_i] \begin{bmatrix}
		\mathcal U^{0,\mathcal{T}}_i \\
		\mathcal N^{0,\mathcal{T}}_i 
	\end{bmatrix} = {\mathcal X}^{1,\mathcal{T}_i} - D_i\mathcal W^{0,\mathcal{T}}_i - \Gamma_i.
\end{align}
By designing the controller matrix $\mathcal{F}_i(x_i) = \mathcal{U}^{0,\mathcal{T}}_i \mathcal{S}_i(x_i)$, and utilizing conditions~\eqref{martrixS} and~\eqref{Transformation}, we have 
\begin{align*}
	A_i &\mathcal{M}_i\left(x_i\right)+B_i \nu_i 
	=A_i \Upsilon_i(x_i) x_i+B_i \nu_i\\ 
	&=\left(A_i \Upsilon_i(x_i) +B_i \mathcal{F}_i(x_i) \right) x_i= [B_i\quad A_i] \begin{bmatrix}
		\mathcal{F}_i(x_i)\\
		\Upsilon_i(x_i)
	\end{bmatrix} x_i\\ & = [B_i~ A_i] \! \! \begin{bmatrix}
		\mathcal U^{0,\mathcal{T}}_i \\
		\mathcal N^{0,\mathcal{T}}_i 
	\end{bmatrix}\!\! \mathcal {S}_i (x_i) x_i \!\overset{\eqref{data3}}{=} \! \! ({\mathcal X}^{1,\mathcal{T}}_i \! \!\!  -\!\!  D_i\mathcal W^{0,\mathcal{T}}_i \! \!\!  -\!  \Gamma_i)\mathcal {S}_i (x_i)x_i.
\end{align*}
Therefore, $(\mathcal X^{1,\mathcal{T}}_i -D_i\mathcal W^{0, \mathcal{T}}_i - \Gamma_i)\mathcal {S}_i (x_i)$ is the data-based representation of $A_i \Upsilon_i(x_i) + B_i\mathcal{F}_i(x_i)$, which completes the proof. 
\end{proof} 

\begin{remark}\label{Rank-condition}
To potentially guarantee that the matrix $\mathcal N^{0,\mathcal{T}}_i$ maintains full row rank, the number of samples, denoted as $\mathcal{T}$, must be at least equal to $N_i +1$. This condition is readily verifiable since  $\mathcal N^{0,\mathcal{T}}_i$ is constructed from sampled data.
\end{remark}
\begin{remark}
		If the data is noise-free, similar to \cite{nejati2022data}, identifying subsystems' matrices utilizing the least squares would be possible. However, our data-driven framework accommodates \emph{unknown-but-bounded} noise in the state derivative data, which makes it no longer straightforward to exactly identify the matrices of subsystems without any error. As these matrices are not identified in our framework, there is no need to handle the associated non-zero error.
\end{remark}

Leveraging the data-based representation outlined in Lemma~\ref{Lemma1}, we offer the following theorem, as one of the main contributions of our work, to design a CSBC and its decentralized controller for each subsystem solely based on data.
\begin{theorem}\label{Thm:main}
	Consider an unknown ct-NPS $\Theta_i$, with its data-based representation $A_i\Upsilon_i(x_i) + B_i\mathcal{F}_i(x_i) = (\mathcal X_i^{1,\mathcal{T}} - D_i\mathcal W_i^{0,\mathcal{T}} - \Gamma_i)\mathcal {S}_i (x_i)$, according to Lemma \ref{Lemma1}. Suppose there exist a  polynomial matrix
	$\mathcal H_i(x_i) \in \mathbb R^{\mathcal{T} \times n_i}$ such that
	\begin{align}\label{matrixP}
		& \mathcal N_i^{0,\mathcal{T}}\mathcal H_i(x_i) =\Upsilon_i(x_i) P_i^{-1},\:\text{with} ~ P_i \succ 0.
	\end{align}
	Let there exist constants $\phi_i,  \gamma_i,\beta_i,\varepsilon_i,\pi_i, \mu_i\in\mathbb{R}^{+}$ so that the following conditions are satisfied
 \begin{subequations}\label{New1}
	\begin{align}\label{con}
		&x_i^\top \big[\Upsilon_i^\dagger(x_i)\mathcal N_i^{0,\mathcal{T}}\mathcal H_i(x_i)\big]^{-1}  x_i \geq \phi_i\|x_i\|^2, ~\forall x \in X_{i},\\\label{con1}
		&x_i^\top \big[\Upsilon_i^\dagger(x_i)\mathcal N_i^{0,\mathcal{T}}\mathcal H_i(x_i)\big]^{-1} x_i \leq \gamma_i,\quad\quad~\: \forall x_i \in X_{0_i},\\\label{con2}
		&x_i^\top \big[\Upsilon_i^\dagger(x_i)\mathcal N_i^{0,\mathcal{T}}\mathcal H_i(x_i)\big]^{-1} x_i\geq \beta_i, \quad\quad~\: \forall x_i \in X_{u_i},\\\label{con11}
		&\quad\:\: \begin{bmatrix}
		 \mathcal{G}_i & \mathcal{H}_i^\top(x)\\
		 \mathcal{H}_i(x) & \mu_i \mathds{I}_{\mathcal{T}}
		\end{bmatrix} \succeq 0, \quad\quad\quad\quad\quad~\:\:\, \forall x_i \in X_i,
	\end{align}
  \end{subequations}
with 
\begin{align}\notag
 \mathcal{G}_i &=-\varepsilon_i P_i^{-1}-(\mathcal X_i^{1,\mathcal{T}} - D_i\mathcal W_i^{0,\mathcal{T}})\mathcal H_i(x_i)\\ \label{G} &~~~- \mathcal H_i(x_i)^\top(\mathcal X_i^{1,\mathcal{T}} - D_i\mathcal W_i^{0,\mathcal{T}})^\top -\mu_i \Xi_i \Xi_i^\top - \pi_i \mathds{I}_{n_i}.
 \end{align}
 Then $\mathds B_i(x_i) = x_i^\top \big[\Upsilon_i^\dagger(x_i)\mathcal N_i^{0,\mathcal{T}}\mathcal H_i(x_i)\big]^{-1}  x_i$ is a CSBC with $\rho_i = \frac{1}{\pi_i}\Vert D_i \Vert^2$, and $\nu_i =  \mathcal{U}^{0,\mathcal{T}}_i \mathcal H_i(x_i)\big[\Upsilon_i^\dagger(x_i)\mathcal N_i^{0,\mathcal{T}}\mathcal H_i(x_i)\big]^{-1} x_i$ is its decentralized controller.
\end{theorem}

\begin{proof} 
	Given condition~\eqref{matrixP}, it is clear that satisfaction of conditions~\eqref{con}-\eqref{con2} implies conditions~\eqref{subsys1}-\eqref{subsys111} with  $P_i = [{\Upsilon_i^\dagger(x_i)}\mathcal N_i^{0,\mathcal{T}}\mathcal H_i(x_i)]^{-1}$.  We now proceed with showing condition~\eqref{eq:martingale2}, as well. Since ${\Upsilon_i(x_i)}P_i^{-1} =  \mathcal N_i^{0,\mathcal{T}}\mathcal H_i(x_i)$ according to~\eqref{matrixP}, and $P_i^{-1}P_i = \mathds{I}_{n_i}$, then ${\Upsilon_i(x_i)} = \mathcal N_i^{0,\mathcal{T}}\mathcal H_i(x_i)P_i$. Since ${\Upsilon_i(x_i)} = \mathcal N_i^{0,\mathcal{T}}\mathcal {S}_i (x_i)$ as per~\eqref{martrixS}, then one could set $\mathcal {S}_i (x_i) = \mathcal H_i(x_i)P_i$ and, accordingly, $\mathcal {S}_i (x_i) P_i^{-1} = \mathcal H_i(x_i)$. Given that $A_i{\Upsilon_i(x_i)} + B_i\mathcal{F}_i(x_i) = (\mathcal X_i^{1,\mathcal{T}} - D_i\mathcal W_i^{0,\mathcal{T}}-{\Gamma_i})\mathcal {S}_i (x_i)$ according to Lemma~\ref{Lemma1}, we have
\begin{align}\nonumber
	(A_i&\Upsilon_i(x_i) + B_i\mathcal{F}_i(x_i))P_i^{-1}\\\nonumber &= (\mathcal X_i^{1,\mathcal{T}} - D_i\mathcal W_i^{0,\mathcal{T}}-\Gamma_i)\underbrace{\mathcal {S} (x_i)P_i^{-1}}_{\mathcal H_i(x_i)}\\\label{matrix_H} &= (\mathcal X_i^{1,\mathcal{T}} - D_i\mathcal W_i^{0,\mathcal{T}}-\Gamma_i)\mathcal H_i(x_i).
\end{align}
With the definition of Lie derivative in~\eqref{Lie derivative}, we have
\begin{align*}
	\mathcal{L}\mathds B_i(x_i) &=\partial_{x_i} \mathds B_i(x_i)(A_i\mathcal M_i(x_i) + B_i\nu_i + D_iw_i)
	\\ &=2 x_i^\top P_i \big(\!(A_i \Upsilon_i(x_i) \!+\!B_i\mathcal{F}_i(x_i)\!) x_i\!+\! D_i w_i\big) 
		\\ &=2 x_i^\top P_i (A_i \Upsilon_i(x_i) \!+\!B_i\mathcal{F}_i(x_i)) x_i  + 2 \underbrace{x_i^\top P_i}_{a_i} \underbrace{D_i w_i}_{c_i}.
\end{align*}
According to the Cauchy-Schwarz inequality \cite{bhatia1995cauchy}, \emph{i.e.,}  $a_i c_i \leq \Vert a_i \Vert \Vert c_i \Vert,$ for any $a^\top_i, c_i \in \R^{n_i}$, followed by
employing Young's inequality \cite{young1912classes}, \emph{i.e.,} $\Vert a_ i \Vert \Vert c_i \Vert \leq \frac{\pi_i}{2} \Vert a_i \Vert^2 + \frac{1}{2\pi_i} \Vert c_i \Vert^2$, for any $\pi_i \in \mathbb R^+$, one has
\begin{align*}
	\mathcal{L}\mathds B_i(x_i) &\leq  2 x_i^\top P_i (A_i \Upsilon_i(x_i) \!+\!B_i\mathcal{F}_i(x_i)) x_i \\
	&~~~\!+\! \pi_i x_i^\top P_i P_ix_i  + \frac{1}{\pi_i} \Vert D_i \Vert^2 \Vert w_i\Vert^2.
\end{align*}
By expanding the expression given above and factorizing the term $ x_i^\top P_i$ from left and $P_i x_i$ from right, we have
\selectfont\begin{align*}
	\mathcal{L}\mathcal B_i(x_i)&\leq x_i^\top P_i\Big[(A_i\Upsilon_i(x_i) +B_i \mathcal F_i(x_i))P_i^{-1}\\ &~~~+ P_i^{-1} (A_i \Upsilon_i(x_i)+B_i \mathcal F_i(x_i))^\top + \pi_i \mathds{I}_{n_i} \Big] P_i x_i\\&~~~  + \frac{1}{\pi_i}\Vert D_i\Vert^2\Vert w_i \Vert^2.
\end{align*}

\begin{figure*}[t!]
	\rule{\textwidth}{0.1pt}
	\begin{align}
		\nonumber
		\mathcal{L}\mathds{B}(x) & \overset{{\eqref{Compose-result}}}{=} \mathcal{L}\sum_{i=1}^\mathcal{Q} \mathds{B}_i (x_i) = \sum_{i=1}^\mathcal{Q}\mathcal{L}\mathds{B}_i(x_i)  \overset{{\eqref{eq:martingale2}}}{\leq} \sum_{i=1}^\mathcal{Q}(-\varepsilon_i \mathds{B}_i\big(x_i\big)+\rho_i\|w_i\|^2) = \sum_{i=1}^\mathcal{Q}\big(-\varepsilon_i \mathds{B}_i\left(x_i\right)+\rho_i \sum_{j=1, i \neq j}^\mathcal{Q}\|w_{ij}\|^2\big)\\\nonumber
		&\,  \overset{{\eqref{Internal}}}{=} \sum_{i=1}^\mathcal{Q}\big(-\varepsilon_i \mathds{B}_i\left(x_i\right)+\rho_i \sum_{j=1, i \neq j}^\mathcal{Q}{\|x_j\|^2}\big)
		\overset{{\eqref{subsys1}}}{\leq} \sum_{i=1}^\mathcal{Q}\big(-\varepsilon_i \mathds{B}_i(x_i)+\sum_{j=1, i \neq j}^\mathcal{Q} \frac{\rho_i}{\phi_j} \mathds{B}_j\left(x_{j}\right)\big)\\	\label{chain_inequality}
		& \overset{{\eqref{gain}}}{=} \sum_{i=1}^\mathcal{Q}\big(-\varepsilon_i \mathds{B}_i(x_i)+\sum_{j=1, i \neq j}^\mathcal{Q} \hat{\delta}_{ij} \mathds{B}_j\left(x_{j}\right)\big)
		\overset{{\eqref{gain}}}{=}\mathbf{1}_\mathcal{Q}^{\top}(-\hat{\varepsilon}+\Delta)\underbrace{\left[\mathds{B}_1\left( x_1\right) ; \ldots ; \mathds{B}_\mathcal{Q}\left( x_\mathcal{Q}\right)\right]}_{\bar{\mathds{B}}(x)} \overset{{\eqref{Comp-proof}}}{\leq} -\varepsilon  \mathds B (x).
	\end{align}
	\rule{\textwidth}{0.1pt}
\end{figure*}

\noindent Then, by replacing $(A_i \Upsilon_i(x_i) + B_i\mathcal{F}_i(x_i))P_i^{-1}$ with its data-based representation according to~\eqref{matrix_H}, one can obtain
\begin{align}\notag
	\mathcal{L}\mathds B_i(x_i) &\leq  x_i^\top \!\!P_i\Big[(\mathcal X_i^{1,\mathcal{T}} -D_i \mathcal W_i^{0,\mathcal{T}} -{\Gamma_i})\mathcal H_i(x_i) \!\\\notag &~~~+ \mathcal H_i(x_i)^\top\!(\mathcal X_i^{1,\mathcal{T}} \!\!\!-\!\! D_i \mathcal W_i^{0,\mathcal{T}}-{\Gamma_i})^\top
	+ {\pi_i \mathds{I}_{n_i}} \Big]P_i x_i\\\label{Youngs-ex} & ~~~+ \frac{1}{\pi_i} \Vert D_i \Vert^2\Vert w_i\Vert^2.
\end{align}
According to Young's inequality for matrices $L_i(x_i)$ and $M_i(x_i)$ \cite{Ando1995}, which states that 
$L_i(x_i) M_i(x_i)^\top + M_i(x_i) L_i(x_i)^\top \succeq -\mu_i L_i(x_i) L_i(x_i)^\top - \mu^{-1}_i M_i(x_i) M_i(x_i)^\top$ 
for any scalar $\mu_i > 0$, and by applying this to 
$ - \Gamma_i \mathcal{H}_i(x_i)  - \mathcal{H}_i(x_i)^\top \Gamma_i^\top$ in \eqref{Youngs-ex}, we obtain
\begin{align*}
	\mathcal{L}\mathds B_i(x_i) &\leq  x_i^\top \!\!P_i\Big[(\mathcal X_i^{1,\mathcal{T}} -D_i \mathcal W_i^{0,\mathcal{T}} \!)\mathcal H_i(x_i) +{\mu_i \Gamma_i \Gamma_i^\top}\!\\ &~~~+ {\frac{1}{\mu_i}\mathcal{H}_i(x_i)^\top \mathcal{H}_i(x_i)}\!+\!\mathcal H_i(x_i)^\top\!(\mathcal X_i^{1,\mathcal{T}} \!\!\!-\!\! D_i \mathcal W_i^{0,\mathcal{T}})^\top
	\\ &~~~+ \pi_i \mathds{I}_{n_i} \Big]P_i x_i + \frac{1}{\pi_i} \Vert D_i \Vert^2\Vert w_i\Vert^2.
\end{align*}
Then, according to the inequality in~\eqref{noise-bound}, one has
\begin{align*}
	\mathcal{L}\mathds B_i(x_i) &\leq  x_i^\top \!\!P_i\Big[(\mathcal X_i^{1,\mathcal{T}} -D_i \mathcal W_i^{0,\mathcal{T}} \!)\mathcal H_i(x_i) +{\mu_i \Xi_i \Xi_i^\top}\!\\ &~~~+ {\frac{1}{\mu_i}\mathcal{H}_i(x_i)^\top \mathcal{H}_i(x_i)} \!\!+\!\mathcal H_i(x_i)^\top\!(\mathcal X_i^{1,\mathcal{T}} \!\!\!-\!\! D_i \mathcal W_i^{0,\mathcal{T}})^\top
	\\ &~~~+ \pi_i \mathds{I}_{n_i} \Big]P_i x_i + \frac{1}{\pi_i} \Vert D_i \Vert^2\Vert w_i\Vert^2.
\end{align*}
According to the Schur complement \cite{zhang2006schur}, we can conclude that the satisfaction of condition \eqref{con11} results in
\begin{align*}
	\mathcal{L}\mathds B_i(x_i)&\leq - \varepsilon_i x_i^\top P_i\overbrace{\underbrace{[{\Upsilon_i^\dagger(x_i)}\mathcal N_i^{0,\mathcal T}\mathcal H_i(x_i)]}_{P_i^{-1}} P_i}^{\mathds I_{n_i}} x_i\\ & ~~~ +  \frac{1}{\pi_i} \Vert D_i \Vert^2\Vert w_i\Vert^2 = - \varepsilon_i \mathds B_i(x_i) + \rho_i\Vert w_i\Vert^2,
\end{align*}
fulfilling condition~\eqref{eq:martingale2} with $\rho_i = \frac{1}{\pi_i}\Vert D_i \Vert^2$.  Hence, $\mathds B_i(x_i) = x_i^\top \big[{\Upsilon_i^\dagger(x_i)}\mathcal N_i^{0,\mathcal{T}}\mathcal H_i(x_i)\big]^{-1}  x_i$ is a CSBC and $\nu_i =  \mathcal{U}^{0,\mathcal{T}}_i \mathcal H_i(x_i)\big[{\Upsilon_i^\dagger(x_i)}\mathcal N_i^{0,\mathcal{T}}\mathcal H_i(x_i)\big]^{-1} x_i$ is its decentralized controller, thereby concluding the proof.
\end{proof} 

\begin{remark}
Our framework can also accommodate exogenous disturbances and model mismatches in the same manner as noisy state derivative data. To do so, one should assume these uncertainties are unknown but bounded, or in the case of model mismatches, bounded with a known bound, consistent with prior work. Nevertheless, incorporating multiple sources of uncertainty increases conservativeness, as all corresponding bounds appear in the proposed condition~\eqref{con11}.
\end{remark}
	
\subsection{Computation of CSBC and Local Controllers}
We reformulate conditions \eqref{con}-\eqref{con11} as a sum-of-squares (SOS) optimization program, aimed at designing a CSBC $\mathds{B}_i$, along with its associated control polynomial matrix $\mathcal{F}_i(x_i)$ for each subsystem $\Theta_i$, as presented in the following lemma.

\begin{lemma}\label{SOS}
Consider sets $X_i$, $X_{0_i}$, and $X_{u_i}$, each defined by vectors of polynomial inequalities in the form of $X_i = \{x_i \in \mathbb{R}^{n_i} |~ b_i(x_i) \geq 0\}$, $X_{0_i} = \{x_i \in \mathbb{R}^{n_i} |~ b_{0_i}(x_i) \geq 0\}$, and $X_{u_i} = \{x_i \in \mathbb{R}^{n_i} |~ b_{u_i}(x_i) \geq 0\}$, respectively. If there exist a polynomial matrix $\mathcal{H}_i(x_i)$,  constants $\phi_i,\gamma_i,\beta_i, \varepsilon_i, \pi_i, {\mu_i}\in \mathbb{R}^{+}$, and vectors of SOS polynomials $\lambda_{0_i}(x_i)$, $\lambda_{u_i}(x_i)$, $\lambda_i(x_i)$, such that 
	\begin{subequations}\label{SOS12}
		\begin{align}\label{SOS0}
			& x_i^\top \big[{\Upsilon_i^\dagger(x_i)}\mathcal N_i^{0,\mathcal{T}}\mathcal H_i(x_i)\big]^{-1}\!\! x_i \! - \! \lambda_i^\top \! \! \left(x_i\right) b_i \! \left(x_i\right)-\phi_i\left\|x_i\right\|^2 ~{\geq}~ 0,\\\label{SOS1}
			 &-x_i^\top \big[{\Upsilon_i^\dagger(x_i)}\mathcal N_i^{0,\mathcal{T}}\mathcal H_i(x_i)\big]^{-1} \! x_i\!-\lambda_{0_i}^\top\!\left(x_i\right) b_{0_i}\left(x_i\right)\!+\gamma_i ~{\geq} ~ 0
		\end{align}
				\begin{align}\label{SOS2}
			& \quad\:x_i^\top \big[{\Upsilon_i^\dagger(x_i)}\mathcal N_i^{0,\mathcal{T}}\mathcal H_i(x_i)\big]^{-1} \! x_i\!-\lambda_{u_i}^\top\left(x_i\right) b_{u_i}\!\left(x_i\right)\!-\beta_i ~ {\geq} ~ 0,\\\label{SOS3}
			& \quad\quad\: {\begin{bmatrix}
					\mathcal{G}_i & \mathcal{H}_i^\top(x)\\
					\mathcal{H}_i(x) & \mu_i \mathds{I}_{\mathcal{T}}
				\end{bmatrix}-(\lambda_i^\top\left(x_i\right) b_i\left(x_i\right))\mathds{I}_{(n_i + \mathcal{T})} ~{\succeq} ~0},
		\end{align}
	\end{subequations}
	with $\mathcal{G}_i$ as in~\eqref{G}, are SOS polynomials while condition~\eqref{matrixP} is also satisfied, then $\mathds B_i(x_i) =  x_i^\top \big[{\Upsilon_i^\dagger(x_i)}\mathcal N_i^{0,\mathcal{T}}\mathcal H_i(x_i)\big]^{-1} x_i$ is a CSBC and $\nu_i =  \mathcal{U}^{0,\mathcal{T}}_i \mathcal H_i(x_i)\big[{\Upsilon_i^\dagger(x_i)}\mathcal N_i^{0,\mathcal{T}}\mathcal H_i(x_i)\big]^{-1} x_i$ is its decentralized controller.
\end{lemma}

\begin{remark}\label{asset}
	One can employ existing software tools in the literature, like \textsf{SOSTOOLS} \cite{SOSTOOLS} alongside semi-definite programming (SDP) solvers such as \textsf{Mosek} \cite{mosek}, for the effective enforcement of conditions~\eqref{SOS0}-\eqref{SOS3}.
\end{remark}

\begin{remark}
The term $\varepsilon_i P_i^{-1}$ in condition~\eqref{SOS3} introduces bilinearity, as both $\varepsilon_i$ and $P_i$ are decision variables. To avoid this, we fix the scalar $\varepsilon_i$ a priori in our implementations. In fact, as $\varepsilon_i$ is scalar, the bilinear problem can be handled via a bisection (line search) over $\varepsilon_i$ to find the smallest feasible value. It should be noted that the coefficients $\pi_i$ and $\mu_i$ are decision variables and can be computed by the solver.
\end{remark}

\section{Compositional Construction of CBC for Interconnected Network}\label{compositional}
In  this section, we offer our compositional framework for designing a CBC and its associated safety controller across interconnected networks, via CSBC of subsystems, constructed from data.
In the next theorem, we analyze network $\Theta =\mathcal{I}$ $\left(\Theta_1, \ldots, \Theta_\mathcal{Q}\right)$ by driving \emph{small-gain} compositional reasoning to construct a CBC for an interconnected network based on CSBC of its individual subsystems. Before presenting the main compositionality result of the work, we define
\begin{equation}\label{gain}
\hat\varepsilon:=\operatorname{diag}\left({\varepsilon}_1, \ldots,{\varepsilon}_\mathcal{Q}\right), \; \Delta:=\{\hat{\delta}_{i j}\} ~\text{with}~ \hat{\delta}_{i j}=\frac{\rho_i}{\phi_j},
\end{equation}
where $\hat{\delta}_{i i}=0$, for any $i \in\{1, \ldots, \mathcal{Q}\}$.
\begin{theorem}\label{Small}
Consider a network $\Theta=\mathcal{I}$ $\left(\Theta_1, \ldots, \Theta_\mathcal{Q}\right)$, induced by $\mathcal{Q} \in \mathbb{N}^{+}$ subsystems $\Theta_i$. Suppose that each $\Theta_i$ admits a CSBC $\mathds{B}_i$, derived from data, according to  Theorem~\ref{Thm:main}. If
\begin{subequations} 
\begin{align} \label{compose1}
\quad\quad\quad\quad\quad\quad\quad&\sum_{i=1}^\mathcal{Q} \beta_i>\sum_{i=1}^\mathcal{Q} \gamma_i,\\\nonumber
\!\!\!\!\!\!\mathbf{1}_\mathcal{Q}^{\top}(-\hat{\varepsilon}+\Delta) &=:\left[\varpi_1 ; \ldots ; \varpi_\mathcal{Q}\right]^{\top}<0, ~\text {equivalently, }~\\\label{compose2}  \varpi_i & <0, \forall i \in\{1, \ldots, \mathcal{Q}\},
\end{align}
\end{subequations} 
then
\begin{align}\nonumber
\mathds{B}(x)&:=\sum_{i=1}^\mathcal{Q} \mathds{B}_i\left( x_i\right)\\\label{Compose-result} & = \sum_{i=1}^\mathcal{Q}  x_i^\top \big[{\Upsilon_i^\dagger(x_i)}\mathcal N_i^{0,\mathcal{T}}\mathcal H_i(x_i)\big]^{-1} x_i,
\end{align}
is a CBC for the network $\Theta$ with $\gamma:=\sum_{i=1}^\mathcal{Q} \gamma_i,  \beta:=\sum_{i=1}^\mathcal{Q} \beta_i,$ $\varepsilon:=-\varpi $, with $\max _{1 \leq i \leq \mathcal{Q}} \varpi_i<\varpi<0$. In addition, $ \nu = [\nu_1;\dots; \nu_\mathcal{Q}]$ with $\nu_i =  \mathcal{U}^{0,\mathcal{T}}_i \mathcal H_i(x_i)\big[{\Upsilon_i^\dagger(x_i)}\mathcal N_i^{0,\mathcal{T}}\mathcal H_i(x_i)\big]^{-1} x_i, i\in\{1,\dots, \mathcal{Q}\}$, is the safety controller of the network. 
\end{theorem}

\begin{proof} 
We first show that conditions~\eqref{CBC1}, \eqref{CBC2} hold. For any ${x:=\left[x_1 ; \ldots ; x_\mathcal{Q}\right]} \in X_0=\prod_{i=1}^\mathcal{Q} X_{0_i}$, by employing condition \eqref{subsys11}, we have
\begin{equation*}
	\mathds{B}(x)=\sum_{i=1}^\mathcal{Q} \mathds{B}_i\left( x_i\right) \leq \sum_{i=1}^\mathcal{Q} \gamma_i=\gamma,
\end{equation*}
and similarly for any ${x:=\left[x_1 ; \ldots ; x_\mathcal{Q}\right]} \in X_{u}=\prod_{i=1}^\mathcal{Q} X_{u_i}$, utilizing condition \eqref{subsys111}, one has
\begin{equation*}
	\mathds{B}(x)=\sum_{i=1}^\mathcal{Q} \mathds{B}_i\left(x_i\right) \geq \sum_{i=1}^\mathcal{Q} \beta_i=\beta,  
\end{equation*}
satisfying conditions~\eqref{CBC1}, \eqref{CBC2} with $\gamma=\sum_{i=1}^\mathcal{Q} \gamma_i$ and $\beta=\sum_{i=1}^\mathcal{Q} \beta_i$. Note that $\beta>\gamma$ according to condition~\eqref{compose1}.

Now, we show that condition~\eqref{cbceq} holds, as well. By employing conditions~\eqref{subsys1}, \eqref{eq:martingale2}, compositionality condition $\mathbf{1}_\mathcal{Q}^{\top}(-\hat{\varepsilon}+$ $\Delta)<0$, and by defining
\begin{equation}\label{Comp-proof}
	\begin{aligned}
		-\varepsilon s & :=\max \left\{\mathbf{1}_\mathcal{Q}^{\top}(-\hat{\varepsilon}+\Delta) \bar{\mathds{B}}(x) \mid \mathbf{1}_\mathcal{Q}^{\top} \bar{\mathds{B}}(x)=s\right\}\!,
	\end{aligned}
\end{equation}
where $\bar{\mathds{B}}(x)=\left[\mathds{B}_1\left( x_1\right) ; \ldots ; \mathds{B}_\mathcal{Q}\left(x_\mathcal{Q}\right)\right]$, \emph{i.e.,} $s=\mathds{B}(x)$, one can obtain the chain of inequalities in~\eqref{chain_inequality}.  

\noindent As the last stage of the proof, we now show that $\varepsilon > 0$ with $\varepsilon=-\varpi $. Since $\mathbf{1}_\mathcal{Q}^{\top}(-\hat\varepsilon+\Delta):=\left[\varpi_1 ; \ldots ; \varpi_\mathcal{Q}\right]^{\top}<0$, and $\max _{1 \leq i \leq \mathcal{Q}} \varpi_i<\varpi <0$ according to \eqref{compose2}, one has
\begin{equation*}
	\begin{aligned}
		-\varepsilon s&\! \overset{{\eqref{Comp-proof}}}{=}\mathbf{1}_{\mathcal{Q}}^{\top}(-\hat\varepsilon+\Delta) \bar{\mathds{B}}(x)\\ &\! \! \overset{{\eqref{compose2}}}{=}\left[\varpi_1 ; \ldots ; \varpi_\mathcal{Q}\right]^{\top}\left[\mathds{B}_1\left(x_1\right)\right.
		\left.\ldots ; \mathds{B}_\mathcal{Q}\left(x_\mathcal{Q}\right)\right]\\ & =\varpi_1 \mathds{B}_1\left(x_1\right)+\ldots+\varpi_\mathcal{Q}\mathds{B}_\mathcal{Q}\left(x_\mathcal{Q}\right)
		\\ &\leq \varpi \left(\mathds{B}_1\left(x_1\right)+\ldots+\mathds{B}_\mathcal{Q}\left(x_\mathcal{Q}\right)\right)\overset{\eqref{Comp-proof}}{=}\varpi  s.
	\end{aligned}
\end{equation*}
Then, $-\varepsilon s \leq \varpi  s$, $\forall s \in \mathbb{R}^+$, and accordingly, $-\varepsilon \leq \varpi $ as $s$ is positive. Since $\max _{1 \leq i \leq \mathcal{Q}} \varpi_i<\varpi <0$, then $\varepsilon=-\varpi >0$, implying that condition~~\eqref{cbceq} is also satisfied, which completes the proof. 
\end{proof} 

We provide Algorithm~\ref{Alg:1} to outline the required procedures for designing a CBC and its associated safety controller.

\begin{remark}
Note that the compositional conditions~\eqref{compose1} and~\eqref{compose2} should be verified a posteriori. For condition~\eqref{compose1}, in homogeneous networks with identical subsystems, satisfying $\beta_i > \gamma_i$ for a single subsystem suffices to ensure the condition holds for the entire network. Condition~\eqref{compose2} is a \emph{sufficient} condition and should also be checked after solving conditions~\eqref{SOS0}–\eqref{SOS3} using data.
\end{remark}

\begin{remark}
While the compositional framework in this work is inspired by model-based approaches, it differs fundamentally as it relies purely on data. Particularly, the interaction gains $\rho_i$ and $\phi_j$ are derived from Theorem~\ref{Thm:main} using only noisy data from subsystems, rather than known models. These gains play a critical role in the compositional condition~\eqref{compose2} of Theorem~\ref{Small} via their appearance in $\Delta$. As a result, our compositional condition is inherently data-driven, making the problem more challenging than in model-based settings where full system models are available.
\end{remark}

\begin{algorithm}[t!]
	\caption{Data-driven design of CBC and its safety controller across unknown network}\label{Alg:1}
\begin{algorithmic}[1]
		\REQUIRE Initial and unsafe sets $X_{0_i}, X_{u_i}$, a dictionary of $\mathcal M_i(x_i)$ up to a known degree,  matrix $\Xi_i$
		\FOR{$i = 1, \cdots, \mathcal{Q}$}
		\STATE Collect input-state trajectories $\mathcal{U}_i^{0,\mathcal{T}}$, $\mathcal{W}_i^{0,\mathcal{T}}$, $\mathcal{X}_i^{0,\mathcal{T}}$, $\mathcal{X}_i^{1,\mathcal{T}}$ as in~\eqref{New}, where $\mathcal{X}_i^{1,\mathcal{T}}:= \bar{\mathcal{X}}_i^{1,\mathcal{T}} + \Gamma_i$
		\STATE Construct $\mathcal{N}^{0,\mathcal{T}}_i$ in~\eqref{trajN} based on $\mathcal M_i(x_i) = \Upsilon_i(x_i)x_i$ in~\eqref{Transformation}
	     \STATE Employ \textsf{SOSTOOLS} to satisfy condition $\Upsilon_i^\dagger(x_i)\mathcal{N}_i^{0,\mathcal{T}}\mathcal{H}_i(x_i) = \mathcal C_i= P_i^{-1}$ in~\eqref{matrixP}, where $\mathcal C_i \succ 0$, and condition~\eqref{SOS3} 
		\STATE Under the constructed $\mathds B_i(x_i) = x_i^\top \big[\Upsilon_i^\dagger(x_i)\mathcal N_i^{0,\mathcal{T}}\mathcal H_i(x_i)\big]^{-1}  x_i$, employ \textsf{SOSTOOLS} and design $\phi_i$, $\gamma_i$ and $\beta_i$ as per conditions~\eqref{SOS0}-\eqref{SOS2}
		\STATE Compute interaction gain $\rho_i = \frac{1}{\pi_i}\Vert D_i \Vert^2$
		\ENDFOR
		\IF{compositionality conditions in~\eqref{compose1}, \eqref{compose2} are satisfied}
		\STATE $\mathds{B}(x) := \sum_{i=1}^\mathcal{Q} \mathds{B}_i(x_i)$, with $\mathds B_i(x_i) = x_i^\top \big[{\Upsilon_i^\dagger(x_i)}\mathcal N_i^{0,\mathcal{T}}\mathcal H_i(x_i)\big]^{-1}  x_i$, is a CBC for network $\Theta$ with its safety controller $ \nu = [\nu_1;\dots; \nu_\mathcal{Q}]$, where $\nu_i =  \mathcal{U}^{0,\mathcal{T}}_i \mathcal H_i(x_i)\big[{\Upsilon_i^\dagger(x_i)}\mathcal N_i^{0,\mathcal{T}}\mathcal H_i(x_i)\big]^{-1} x_i$
		\ELSE 
		\STATE Return to Step~2 and carry out the required procedures again with an expanded dataset containing a greater number of collected samples
		\ENDIF
		\ENSURE CBC $\mathds{B}(x)$ and its safety controller $\nu=[\nu_1; \cdots; \nu_{\mathcal{Q}}]$
	\end{algorithmic}
\end{algorithm}

\begin{table*}[t!]
	\centering
	\caption{Overview of our data-driven results for a set of \emph{homogeneous} interconnected networks. Here, $\mathcal{Q}$ is the number of subsystems, $\mathcal{T}$ the number of samples, and $n_i$ the number of state variables per subsystem. The parameter $\bar{\varkappa}_i$ denotes the noise bound coefficient from~\eqref{noise-bound-2}, while $\gamma_i$ and $\beta_i$ define the initial and unsafe level sets, respectively. The values $\rho_i$ and $\varepsilon_i$ represent the interaction gain and decay rate. We report the running time ($\mathsf{RT}$, in seconds) and memory usage ($\mathsf{MU}$, in megabits) to solve the SOS problem for each subsystem via Algorithm~\ref{Alg:1}. Notably, in strongly connected topologies (\emph{e.g.,} fully interconnected networks), fewer subsystems are considered, as condition~\eqref{compose2} becomes more restrictive due to increased interdependence.}
	\label{tab:system-configurations}
	\resizebox{\linewidth}{!}{\begin{tabular}{@{}llccccccccccc@{}}
			\toprule
			Type of System & Topology & $\mathcal{Q}$ & $\mathcal{T}$&  $n_i$& $\bar{\varkappa}_i$ & $\gamma_i$ & $\beta_i$ & $\rho_i$ & $\varepsilon_i$ & \(\mathsf{RT}\) & $\mathsf{MU}$ \\ 
			\midrule
			\midrule
			\multirow{2}{*}{Lorenz} &Fully connected & $1000$& $15$ & $3$ & $0.03$ &  $ 478.71$ &   $ 490.16$ &  $   4.71 \times 10^{-7}$  &  $0.99$ & $10.49$ & $9.07$ \\
			& Ring & $2000$& $13$  &$3$ & $0.12$ &  $ 501.44$ &    $514.51$  &  $ 1.25  \times 10^{-4}$ & $0.99$  &  $7.32$ & $2.78$ \\
			\hline
			Spacecraft     & Line              & $2000$ & $14$  & $3$ & $0.75$ &   $73.38$  &  $75.10$ &  $0.06$  & $0.99$  & $7.86$  & $2.83$ \\
			\hline
			Lu                     & Star               & $2000$ & $11$  & $3$& $0.04$ &   $729.92$ &   $ 749.85$  &    $1.64 \times 10^{-6}$ &  $ 0.99$& $4.3$  & $2.25$  \\
			\hline
			\multirow{2}{*}{Duffing oscillator}             & Ring              & $2000$ & $20$  & $2$& $0.18$ & $    281.33 $ &  $       291.32 $ &  $   0.01 $  & $0.99$  &  $3.9$  &  $1.67$ \\ 
			& Binary                & $1023$ & $20$ & $2$& $0.08$ &  $ 280.96 $ &  $   290.76$ &     $0.003$ & $0.99$  &  $3.84$  &  $1.67$  \\ 
			\hline
			Chen                   & Fully connected & $1000$ &  $14$  & $3$ &  $0.27$ &   $514.51$ &    $527.59$ &  $1.39  \times 10^{-5}$  &  $0.99$ & $9.98$ &  $8.28$\\  
			\bottomrule
	\end{tabular}}
\end{table*}

\section{Case Study: A Set of Benchmarks}\label{sec:Case_study}

We demonstrate the effectiveness of our proposed approach on a set of benchmark systems, including interconnected networks of identical Lorenz, Chen, Lu~\cite{lopez2019synchronization}, spacecraft~\cite{khalil2002control}, and Duffing oscillator~\cite{WANG2014162} systems (\emph{i.e.}, homogeneous networks), implemented under various interconnection topologies, as well as a \emph{heterogeneous} network interconnected via a line topology (cf. Figure~\ref{Het-fig}). A concise overview of the homogeneous networks is presented in Table~\ref{tab:system-configurations}, while Table~\ref{tab:system-configurations-a} summarizes the key details of the heterogeneous network.
It is worth noting that Lorenz-type systems (\emph{i.e.,} Lorenz, Chen, Lu) are chaotic systems whose networks are widely used in various real-world applications due to their ability to model complex, chaotic behaviors.  This includes applications in secure communications for encryption using chaotic signals\cite{Wang2009}, weather prediction models to simulate atmospheric dynamics \cite{DeterministicNonperiodicFlow}, robotics and autonomous systems to adapt to unpredictable environments \cite{sprott2010elegant}, and neuroscience to model chaotic brain activity and understand neural dynamics and disorders such as epilepsy \cite{strogatz2018nonlinear}.

The main objective in all benchmarks is to design a CBC and its robust safety controller for the interconnected network with an unknown mathematical model, ensuring that the network's states remain within a safe region for an infinite time period. To achieve this, under Algorithm~\ref{Alg:1}, we collect input-state trajectories as in~\eqref{New} and fulfill conditions~\eqref{SOS12} by constructing CSBC and local controllers for all subsystems. Under the compositionality findings of Theorem~\ref{Small}, we then compositionally construct a CBC and its  robust safety controller for interconnected networks. By leveraging Theorem~\ref{Theo1}, we ensure that, under the designed controller from data, all trajectories of unknown interconnected networks starting from $X_0$ remain within the desired safe sets in the infinite time horizon. 

The following subsections provide detailed descriptions of the benchmarks outlined in Tables~\ref{tab:system-configurations} and~\ref{tab:system-configurations-a}. The matrices \(A_i\) and \(B_i\) are treated as unknown in all case studies, while the dictionary of monomials is assumed to be constructed based on the known upper bound on the maximum degree of the monomials appearing in the true subsystem dynamics (cf. Remark~\ref{rem:dict}).
In addition, the computed CSBC and their corresponding local controllers \(\nu_i\) for the subsystems, the CBC and their safety controllers \(\nu\) for the interconnected networks, and network parameters (\emph{i.e.}, \(\gamma\), \(\beta\), and \(\varepsilon\)) are all provided in the Appendix. All simulations were conducted on a \textsf{MacBook M2 chip} with 32~\textsf{GB} of memory. For broader reproducibility, the implementation codes for all case studies are available online.\footnote{\href{https://github.com/OmidAkbarzadeh1991/Data-Driven-CBC-Network}{https://github.com/OmidAkbarzadeh1991/Data-Driven-CBC-Network}}
\subsection{Lorenz Network with Fully-Interconnected  Topology}
 We consider a Lorenz network with a fully-interconnected topology of $1000$ subsystems, with the dynamics of each subsystem described as
\begin{align}\label{lorenz-fully}
		\begin{split}
			& \dot{x}_{i_1}=10 \,x_{i_2}-10\, x_{i_1} - 2 \times 10^{-5}\!\!\!\! \sum_{\substack{j=1, j \neq i}}^{1000} x_{j_1} + \nu_{i_1},\\
			& \dot {x}_{i_2}=28\, x_{i_1}-x_{i_2}-x_{i_1} x_{i_3}- 2 \times 10^{-5}\!\!\!\! \sum_{\substack{j=1, j \neq i}}^{1000} x_{j_2}+\nu_{i_2}, \\
			& \dot{x}_{i_3}=x_{i_1} x_{i_2}-\frac{8}{3}\, x_{i_3}- 2 \times 10^{-5}\!\!\!\! \sum_{\substack{j=1, j \neq i}}^{1000} x_{j_3}+\nu_{i_3}.
		\end{split}
\end{align}
One can rewrite the dynamics in~\eqref{lorenz-fully} in the form of \eqref{sys2} with actual $\mathcal{M}_i(x_i) = [x_{i_1}; x_{i_2}; x_{i_3}; x_{i_1}x_{i_3}; x_{i_1}x_{i_2}]$, and
\begin{align*}
	A_i \! = \!\!\!  \begin{bmatrix} -10& 10 & 0 & 0 & 0  \\28 & -1 & 0 & -1 & 0  \\ 0 & 0 & -\frac{8}{3} & 0 & 1 \end{bmatrix}\!\!, 
	B_i \! = \! \mathds{I}_3, D_{ij} \! = \! -2 \! \times \! 10^{-5} \mathds{I}_{3}.
\end{align*}
It is worth reiterating that the matrices \(A_i\) and \(B_i\) are assumed to be unknown, and the exact form of \(\mathcal{M}_i(x_i)\) is also unavailable. However, a dictionary of monomials up to the known degree of $2$ is given as \begin{align}\label{Dic}
\mathcal{M}_i(x_i)=[x_{i_1}; x_{i_2}; x_{i_3}; x_{i_1}x_{i_3}; x_{i_1}x_{i_2}; x_{i_2}x_{i_3}; x_{i_1}^2; x_{i_2}^2; x_{i_3}^2].
\end{align}
Given the fully-interconnected topology, the matrix \(A(x)\) of the network can be constructed as in~\eqref{A(x)-lorenz-fully}, while the control matrix is \(B = \mathsf{blkdiag}(B_1, \ldots, B_{1000})\). The regions of interest are given by $X_i = [-20,20]^3$, $X_{0_i} = [-3,3]^3$, and $X_{u_i} = [-20,-4]\times[-20,-15]\times[4,20] \cup [8,20]\times[11,20]\times[4,20] \cup [8,20]\times[11,20]\times[-20,-5],$ for all $i \in \{1,\ldots,1000\}$.

\begin{figure*}[t!]
	\centering
	\begin{subfigure}[t]{0.26\textwidth}
		\includegraphics[width=\linewidth]{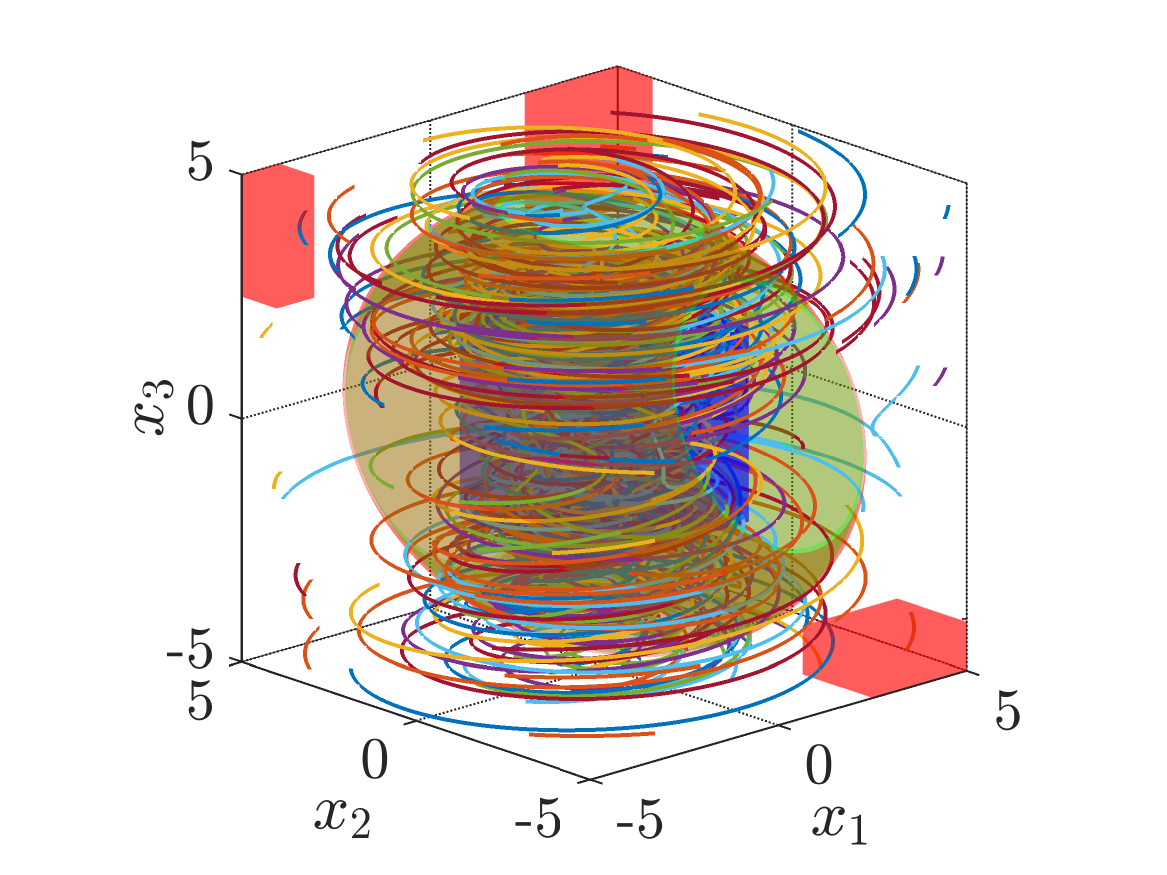}
		\caption{Open-loop trajectories}
	\end{subfigure}
	\hspace{-0.5cm}
	\begin{subfigure}[t]{0.26\textwidth}
		\includegraphics[width=\linewidth]{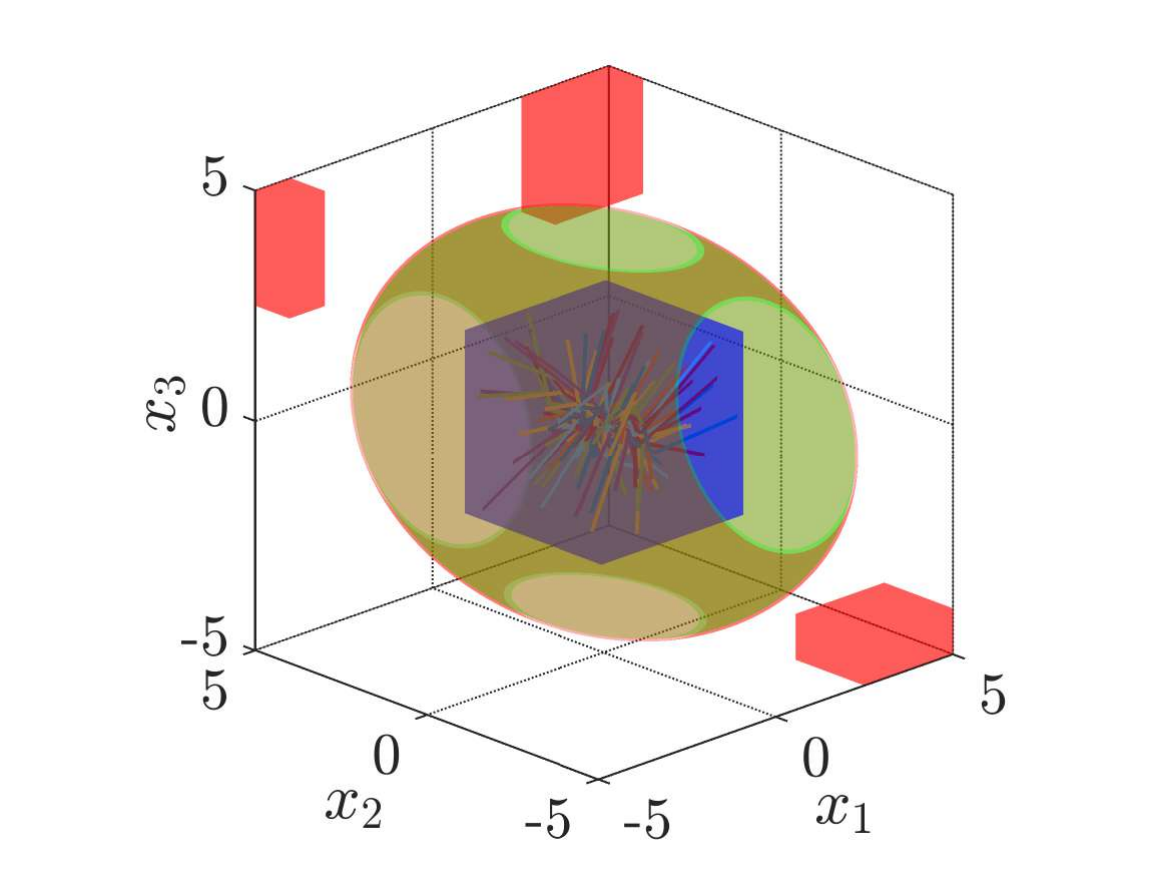}
		\caption{Closed-loop trajectories}
	\end{subfigure}
	\hspace{-0.5cm}	
	\begin{subfigure}[t]{0.26\textwidth}
		\includegraphics[width=\linewidth]{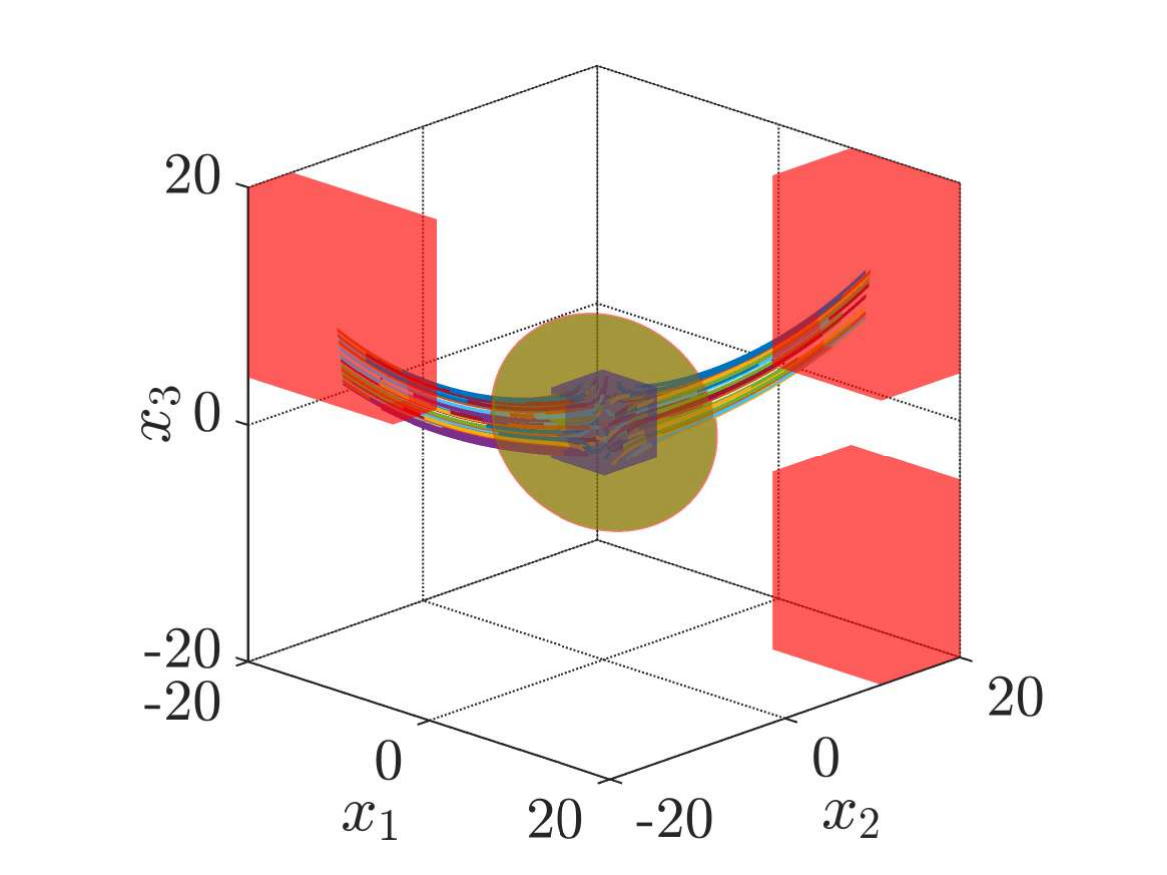}
		\caption{Open-loop trajectories}
	\end{subfigure}
	\hspace{-0.5cm}
	\begin{subfigure}[t]{0.26\textwidth}
		\includegraphics[width=\linewidth]{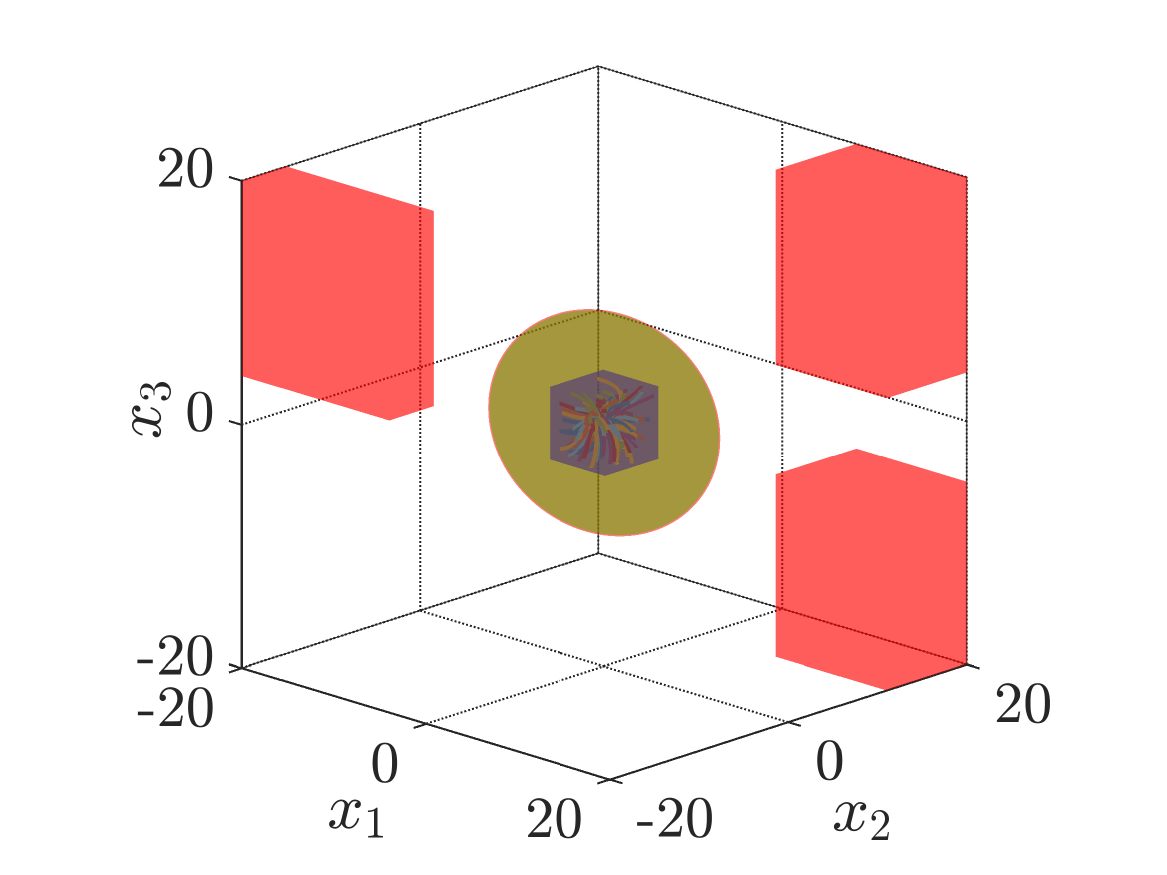}
		\caption{Closed-loop trajectories}
	\end{subfigure}
	\caption{3D trajectories of 120 representative spacecraft subsystems illustrated in subfigures \textbf{(a), (b)}, and Lorenz subsystems in subfigures  \textbf{(c), (d)}, within networks of 2000 and 1000 components, respectively. Subsystems in \textbf{(a)} and \textbf{(b)} are embedded in line interconnection topologies, while those in \textbf{(c)} and \textbf{(d)} are part of fully-interconnected topologies. Green~\protect\greensquare\ and pink~\protect\pinksquare\ surfaces depict the initial and unsafe level sets, respectively. Due to their proximity, the pink unsafe level set may not be clearly visible. Blue~\protect\bluesquare\ and red~\protect\redsquare\ boxes indicate the corresponding initial and unsafe regions. Plots \textbf{(a)} and \textbf{(c)} show open-loop trajectories, illustrating how the subsystems' inherent dynamics quickly lead to violations of safety specifications in the absence of safety controllers.}
	\label{fig:2}
\end{figure*}

\subsection{ Lorenz Network with Ring Interconnection Topology}
We analyze a Lorenz network with a ring topology, with the dynamics of each subsystem, \(i \in \{2, \ldots, 2000\}\), described by
\begin{align}
		\begin{split}\label{lorenz-ring}
			& {\dot{x}_{i_1}=10\, x_{i_2}-10\, x_{i_1} -0.01\,{x}_{(i-1)_1}+\nu_{i_1},}\\
			& {\dot {x}_{i_2}=28\, x_{i_1}-x_{i_2}-x_{i_1} x_{i_3}- 0.01\,{x}_{(i-1)_2} +\nu_{i_2},} \\
			& {\dot{x}_{i_3}=x_{i_1} x_{i_2}-\frac{8}{3}\, x_{i_3}- 0.01\,{x}_{(i-1)_3}+\nu_{i_3},}
		\end{split}
\end{align}
with the first subsystem being affected by the last one. One can rewrite the dynamics in~\eqref{lorenz-ring} in the form of \eqref{sys2} with the actual $\mathcal{M}_i(x_i) = [x_{i_1}; x_{i_2}; x_{i_3}; x_{i_1}x_{i_3}; x_{i_1}x_{i_2}]$ and
\begin{align*}
		A_i = \begin{bmatrix} -10& 10 & 0 & 0 & 0  \\28 & -1 & 0 & -1& 0 \\ 0 & 0 & -\frac{8}{3} & 0 & 1 \end{bmatrix}\!\!, 
		B_i = \mathds{I}_3, D_{ij} &= -0.01 \mathds{I}_{3},
	\end{align*}
where both matrices \(A_i\) and \(B_i\) are treated as unknown. While the exact form of \(\mathcal{M}_i(x_i)\) is unknown, we construct a dictionary of monomials up to the known degree $2$ as in~\eqref{Dic}. The matrices of the interconnected network can be written as 
\begin{align*}
	A(x)~\text{as in ~\eqref{A(x)-lorenz-ring}}, ~~B= \mathsf{blkdiag}(B_1,\ldots,B_{2000}).
\end{align*}
We consider the following regions of interest for all \(i \in \{1, \ldots, 2000\}\): $X_i = [-20, 20]^3$, $X_{0_i} = [-3,3]^3$, and $X_{u_i} = [-20,-10] \times[-20,-5] \times [5, 20] \cup [3.5,20] \times[15,20] \times [5, 20] \cup [3.5,20] \times[15, 20] \times [-20, -5]$.

\subsection{Spacecraft Network with Line Interconnection Topology}
We examine a spacecraft network with a line topology, with each spacecraft dynamics as,  \(i \in \{2, \ldots, 2000\}\),
\begin{align}
	\begin{split}\label{space-line} 
		{\dot{x}_{i_1}} & {= \frac{J_{i_2} - J_{i_3}}{J_{i_1}}\, x_{i_2}\, x_{i_3} 
			+ \frac{1}{J_{i_1}}\, \nu_{i_1} 
			+ \frac{4}{J_{i_1}}\, x_{(i-1)_1},} \\
		{\dot{x}_{i_2}} &{= \frac{J_{i_3} - J_{i_1}}{J_{i_2}}\, x_{i_1}\, x_{i_3} 
			+ \frac{1}{J_{i_2}}\, \nu_{i_2} 
			+ \frac{4}{J_{i_2}}\, x_{(i-1)_2},} \\
		{\dot{x}_{i_3}} &{= \frac{J_{i_1} - J_{i_2}}{J_{i_3}}\, x_{i_1}\, x_{i_2} 
			+ \frac{1}{J_{i_3}}\, \nu_{i_3} 
			+ \frac{4}{J_{i_3}}\, x_{(i-1)_3},}
	\end{split}
\end{align}
while the first subsystem, \emph{i.e.,} $i = 1,$ evolves as
\begin{align*}
	\dot{x}_{i_1} &= \frac{J_{i_2} - J_{i_3}}{J_{i_1}}\, x_{i_2}\, x_{i_3} 
	+ \frac{1}{J_{i_1}}\, \nu_{i_1}, \\
	\dot{x}_{i_2} &= \frac{J_{i_3} - J_{i_1}}{J_{i_2}}\, x_{i_1}\, x_{i_3} 
	+ \frac{1}{J_{i_2}}\, \nu_{i_2}, \\
	\dot{x}_{i_3} &= \frac{J_{i_1} - J_{i_2}}{J_{i_3}}\, x_{i_1}\, x_{i_2} 
	+ \frac{1}{J_{i_3}}\, \nu_{i_3},
\end{align*}
where $u_i=\left[u_{i_1} ; u_{i_2} ; u_{i_3}\right]$ is the torque input, and $J_{i_1}$ to $J_{i_3}$ are the principal moments of inertia. One can rewrite the dynamics given in~\eqref{space-line} in the form of \eqref{sys2} with the actual $\mathcal{M}_i(x_i) = [x_{i_1},x_{i_2},x_{i_3},   x_{i_2}x_{i_3}; x_{i_3}x_{i_1}; x_{i_1}x_{i_2} ]$ and
\begin{align*}
		A_i \!&=\! \begin{bmatrix}  0& 0 & 0 & \frac{J_{i_2} - J_{i_3}}{J_{i_1}}& 0& 0  \\0 & 0& 0 & 0 & \frac{J_{i_3} - J_{i_1}}{J_{i_2}} & 0 \\0& 0 & 0 & 0 & 0 & \frac{J_{i_1} - J_{i_2}}{J_{i_3}} \end{bmatrix}\!\!, \\
		B_i \!&= \! \begin{bmatrix} \frac{1}{J_{i_1}}  & 0  & 0\\ 0 &  \frac{1}{J_{i_2}} &0 \\ 0 &0 &\frac{1}{J_{i_3}}    \end{bmatrix}\!\!,
		D_{ij} =  \begin{bmatrix} \frac{4}{J_{i_1}}  & 0  & 0\\ 0 &  \frac{4}{J_{i_2}} &0 \\ 0 &0 &\frac{4}{J_{i_3}}    \end{bmatrix}\!\!.
			\end{align*}
We construct the dictionary of monomials up to the known degree of $2$ as in \eqref{Dic}. The matrices of the interconnected network can be presented as
\begin{align*}
	A(x)~\text{as in ~\eqref{A(x)-space}}, ~~B= \mathsf{blkdiag}(B_1,\ldots,B_{2000}).
\end{align*}

\begin{figure*}[t!]
		\centering
		\begin{subfigure}[t]{0.26\textwidth}
			\includegraphics[width=\linewidth]{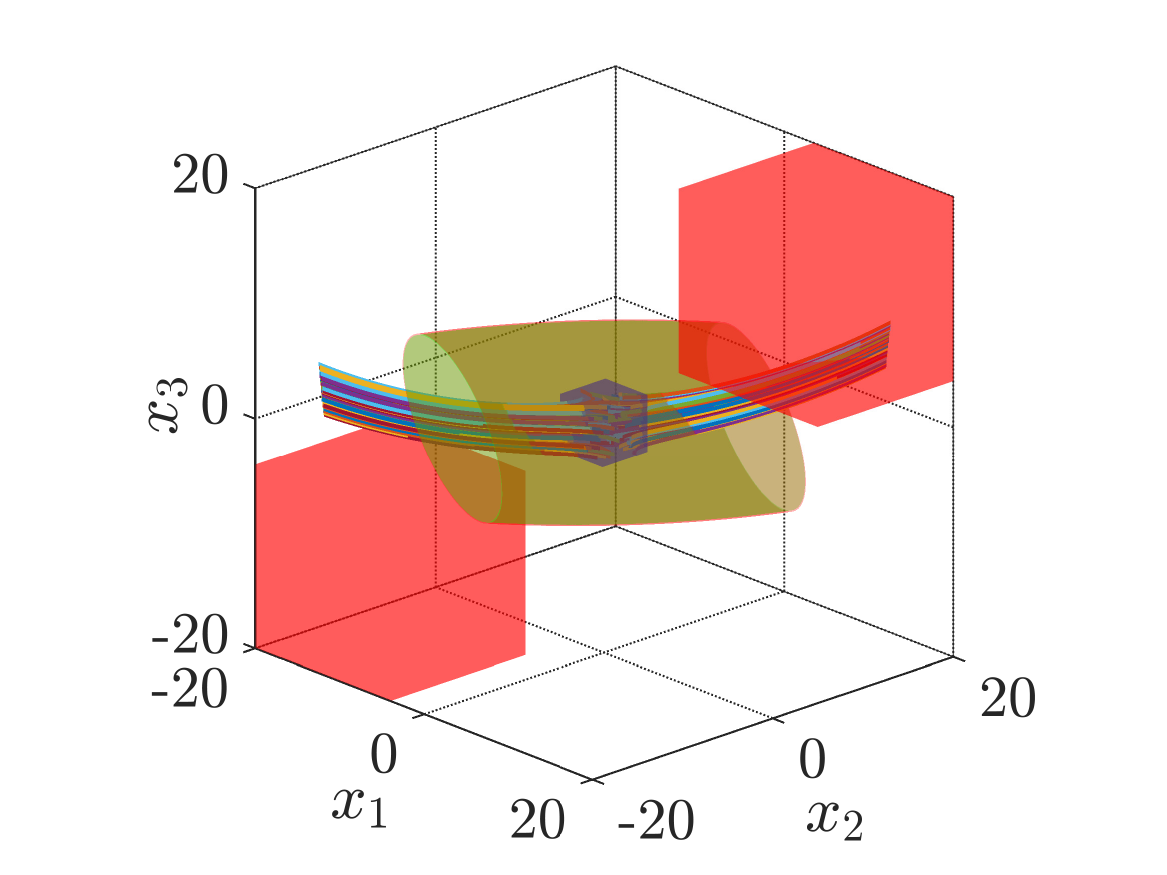}
			\caption{Open-loop trajectories}
		\end{subfigure}
		\hspace{-0.5cm}
		\begin{subfigure}[t]{0.26\textwidth}
			\includegraphics[width=\linewidth]{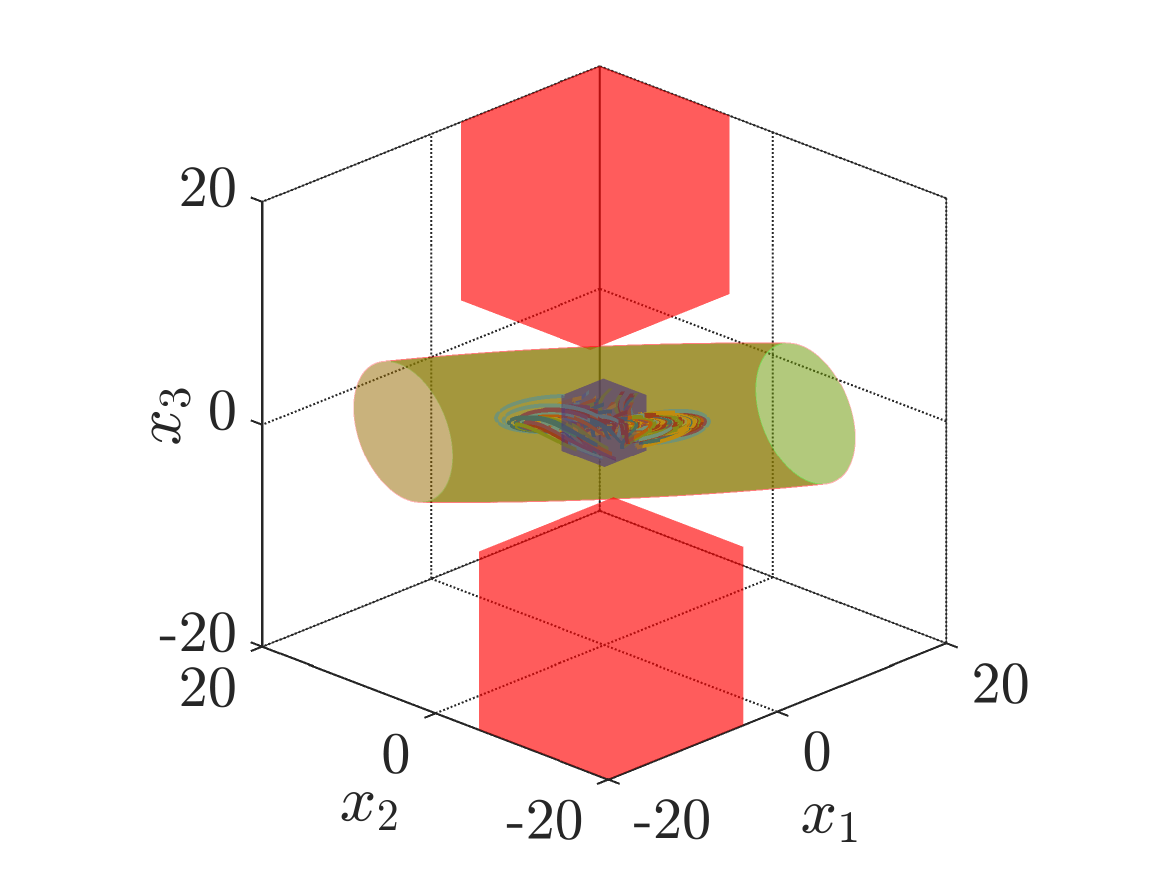}
			\caption{Closed-loop trajectories}
		\end{subfigure}
		\hspace{-0.5cm}
		\begin{subfigure}[t]{0.26\textwidth}
			\includegraphics[width=\linewidth]{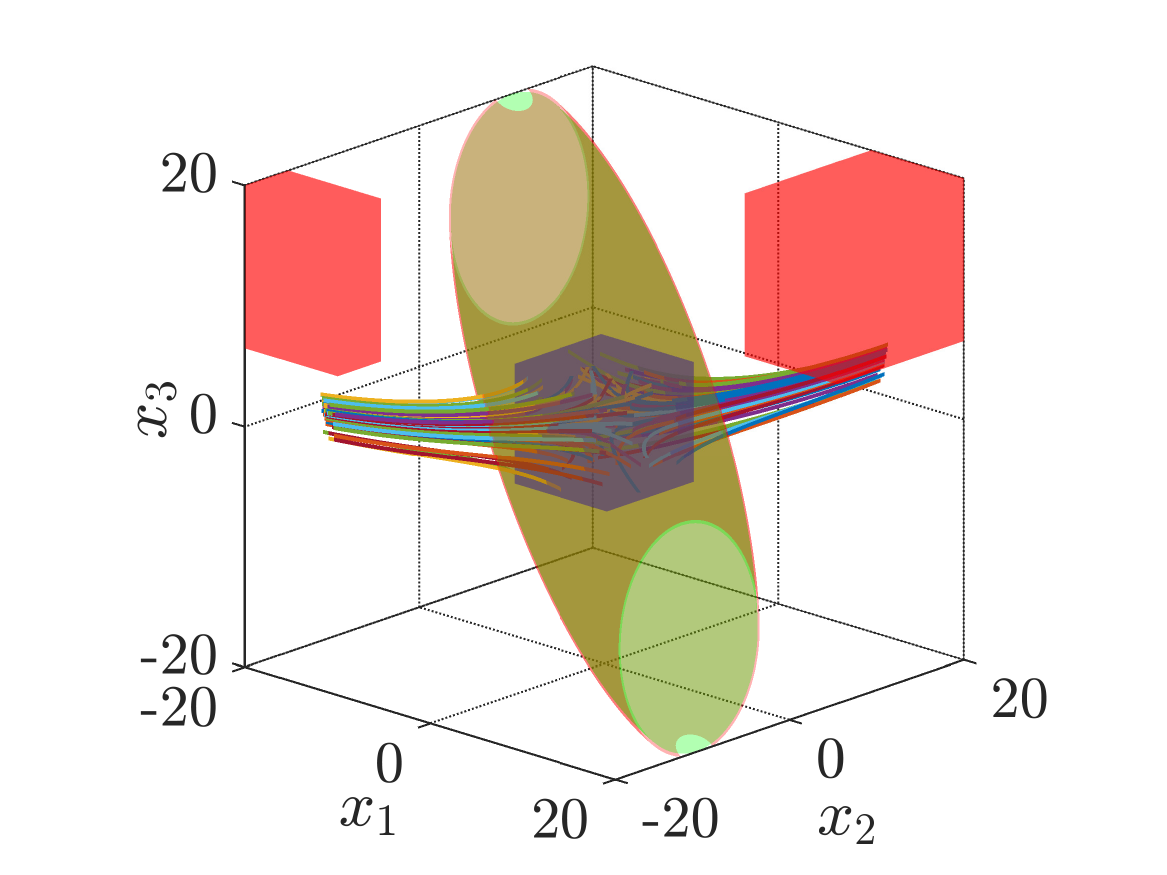}
			\caption{Open-loop trajectories}
		\end{subfigure}
		\hspace{-0.5cm}
		\begin{subfigure}[t]{0.26\textwidth}
			\includegraphics[width=\linewidth]{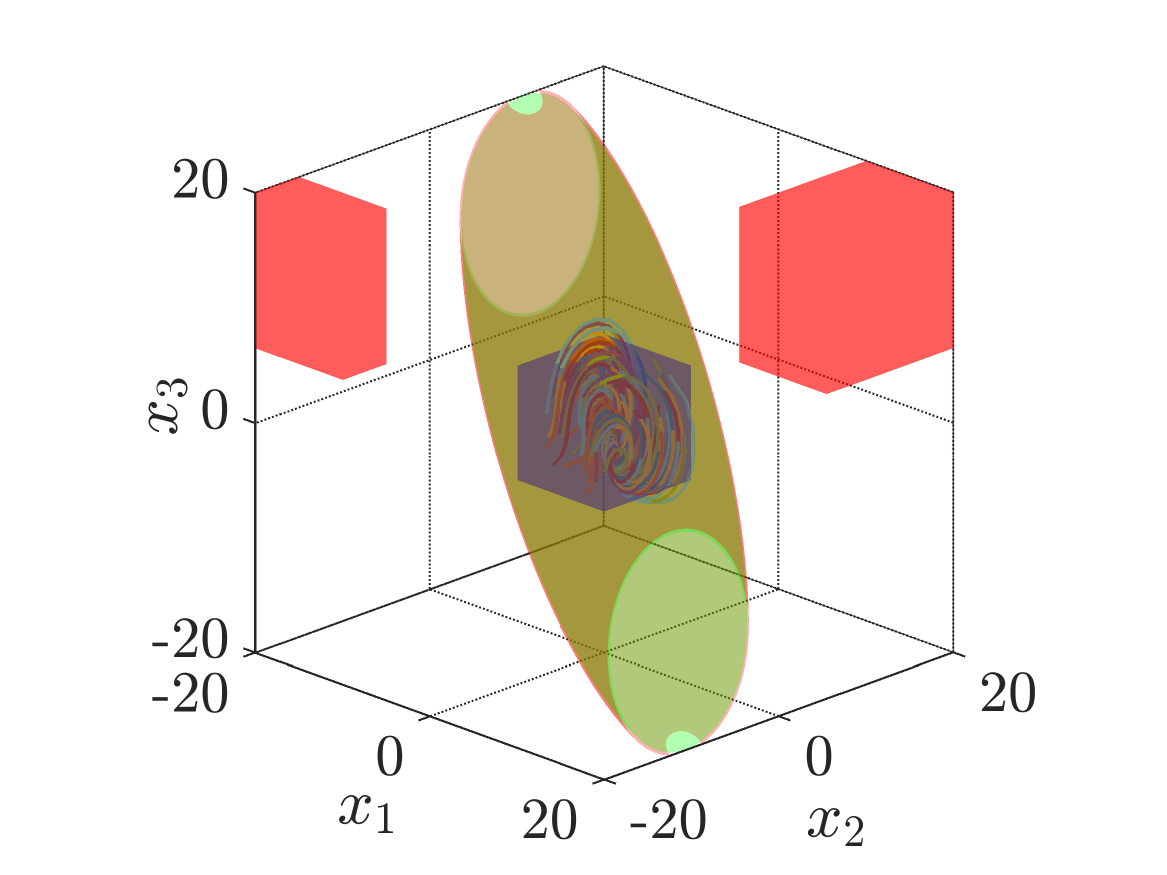}
			\caption{Closed-loop trajectories}
		\end{subfigure}
		\caption{3D trajectories of 120 representative Chen subsystems illustrated in subfigures \textbf{(a), (b)}, and Lu subsystems in subfigures \textbf{(c), (d)}, within networks of $1000$ and 2000 components, embedded in fully connected and star interconnection topologies, respectively. Green~\protect\greensquare\ and pink~\protect\pinksquare\ surfaces represent the initial and unsafe level sets, respectively. Due to their close proximity, the pink unsafe level set may be difficult to distinguish. The corresponding initial and unsafe regions are marked by blue~\protect\bluesquare\ and red~\protect\redsquare\ boxes.}
			\label{fig:4}
\end{figure*}
\begin{figure*}[t!] 
	{\begin{align}\label{A(x)-lorenz-fully}
			A(x) &= \begin{bmatrix}
				A_1 \Upsilon_1(x_1) \!\! & \!\!D_{1\,2} \!\! & \!\!D_{1\,3} \!\! & \!\!\ldots \!\! & \!\!D_{1\,\,1000} \\
				D_{2\,1} \!\! & \!\!A_2 \Upsilon_2(x_2) \!\! & \!\!D_{2\,3} \!\! & \!\!\ldots \!\! & \!\!D_{2\,\,1000} \\
				\vdots \!\! & \!\!\quad \!\! & \!\!\ddots \!\! & \!\!\quad \!\! & \!\!\vdots \\
				D_{999\,\,1} \!\! & \!\!\ldots \!\! & \!\!D_{999\,\,998} \!\! & \!\!A_{999}\Upsilon_{999}(x_{999}) \!\! & \!\!D_{999\,\,1000} \\
				D_{1000\,\,1} \!\! & \!\!\ldots \!\! & \!\!D_{1000\,\,998} \!\! & \!\!D_{1000\,\,999} \!\! & \!\!A_{1000}\Upsilon_{1000}(x_{1000})
			\end{bmatrix}
	\end{align}}
\end{figure*}
\begin{figure*}[t!] 
	{\begin{align}\label{A(x)-lorenz-ring}
			A(x) &= \begin{bmatrix}
				A_1\Upsilon_1(x_1) \!\! \!\! & \!\!\!\!\mathbf{0}_{3 \times 3} \!\! & \!\!\mathbf{0}_{3 \times 3} \!\! & \!\!\cdots \!\! & \!\!D_{1\,\,2000} \\
				D_{2\,1} \!\! & \!\!A_2\Upsilon_2(x_2) \!\! & \!\!\mathbf{0}_{3 \times 3} \!\! & \!\!\cdots \!\! & \!\!\mathbf{0}_{3 \times 3} \\
				\vdots \!\! & \!\! \!\! & \!\!\ddots \!\! & \!\! \!\! & \!\!\vdots \\
				\mathbf{0}_{3 \times 3} \!\! & \!\!\cdots \!\! & \!\!D_{1999\,\,1998} \!\! & \!\!A_{1999}\Upsilon_{1999}(x_{1999}) \!\! & \!\!\mathbf{0}_{3 \times 3} \\
				\mathbf{0}_{3 \times 3} \!\! & \!\!\cdots \!\! & \!\!\mathbf{0}_{3 \times 3} \!\! & \!\!D_{2000\,\,1999} \!\! & \!\!A_{2000}\Upsilon_{2000}(x_{2000})
			\end{bmatrix}
	\end{align}}
\end{figure*}
\begin{figure*}[t!] 
	{\begin{align}\label{A(x)-space}
			A(x) &=  \begin{bmatrix}
				A_1\Upsilon_1(x_1)                 \!\! & \!\!\mathbf{0}_{3 \times 3} \!\! & \!\!\mathbf{0}_{3 \times 3} \!\! & \!\!\cdots \!\! & \!\!\mathbf{0}_{3 \times 3}  \\
				D_{2\,1}     \!\! & \!\!A_2  \Upsilon_2(x_2)                   \!\! & \!\!\mathbf{0}_{3 \times 3} \!\! & \!\!\cdots \!\! & \!\!\mathbf{0}_{3 \times 3} \\
				\!\! & \!\!                       \!\! & \!\!                       \!\! & \!\!       \!\! & \!\!                   \\
				\vdots             \!\! & \!\!                       \!\! & \!\!\ddots                 \!\! & \!\!       \!\! & \!\!\vdots           \\
				\!\! & \!\!                       \!\! & \!\!                       \!\! & \!\!       \!\! & \!\!                   \\
				\mathbf{0}_{3 \times 3} \!\! & \!\!\cdots       \!\! & \!\!D_{1999\,\,1998}     \!\! & \!\!A_{1999}\Upsilon_{1999}(x_{1999})  \!\! & \!\!\mathbf{0}_{3 \times 3} \\
				\mathbf{0}_{3 \times 3} \!\! & \!\!\cdots       \!\! & \!\!\mathbf{0}_{3 \times 3}  \!\! & \!\!D_{2000\,\,1999} \!\! & \!\!A_{2000}\Upsilon_{2000}(x_{2000}) 
			\end{bmatrix}
	\end{align}}
\end{figure*}
The regions of interest are given as follows for all \(i \in \{1, \ldots, 2000\}\): \(X_i = [-5, 5]^3\), \(X_{0_i} = [-2, 2]^3\), and 
	\(
	X_{u_i} = [2.5, 5] \times [-5, -3] \times [-5, -4] \;\cup\; [2.5, 5] \times [4, 5] \times [2.5, 5] \;\cup\; [-5, -4] \times [4, 5] \times [2.5, 5].
	\) The simulation results for representative spacecraft and Lorenz subsystems under line and fully-interconnected topologies are illustrated in Figure~\ref{fig:2}.

\subsection{Lu Network with Star Interconnection Topology}
We analyze a Lu network with a star topology, with subsystem dynamics represented as, \(i \in \{2, \ldots, 2000\}\),
{\begin{align}\label{lu-star}
	\begin{split}
		& \dot{x}_{i_1}=-36\, x_{i_1}+36\, x_{i_2}-10^{-3}\, x_{1_1}, \\
		& \dot{x}_{i_2}=28\, x_{i_2}-x_{i_1} x_{i_3} -10^{-3}\, x_{1_2} +\nu_{i_1}, \\
		& \dot{x}_{i_3}=-20\, x_{i_3}+x_{i_1} x_{i_2}-10^{-3}\, x_{1_3} +\nu_{i_2},
	\end{split}
\end{align}}
while for the first subsystem ($i=1$) the dynamics are
{\begin{align*}
	& \dot{x}_{i_1}=-36\, x_{i_1}+36\, x_{i_2}, \\
	& \dot{x}_{i_2}=28\, x_{i_2}-x_{i_1} x_{i_3}+\nu_{i_1}, \\
	& \dot{x}_{i_3}=-20\, x_{i_3}+x_{i_1} x_{i_2} +\nu_{i_2}.
\end{align*}
One can rewrite the dynamics in~\eqref{lu-star} in the form of \eqref{sys2} with the actual $\mathcal{M}_i(x_i) = [x_{i_1}; x_{i_2}; x_{i_3}; x_{i_1}x_{i_3}; x_{i_1}x_{i_2}]$, and}
{\begin{align*}
		A_i \!\! =\!\!\! \begin{bmatrix} -36& 36 & 0 & 0 & 0 \\0 & 28 & 0 & -1 & 0\\0 & 0 & -20 & 0 & 1  \end{bmatrix}\!\!, B_i \!\! = \!\!\! \begin{bmatrix} 0  & 0 \\ 1 & 0\\ 0 &1  \end{bmatrix}\!\!, D_{ij} \! = \! -10^{-3} \mathds{I}_3.
\end{align*}
Moreover, the dictionary of monomials up to the known degree of $2$ is defined as in \eqref{Dic}. The matrices of the interconnected network can be written as
\begin{align*}
A(x)&=\left\{\mathsf{a}_{ij}\right\}= \begin{cases}A_i \Upsilon_i(x_i),& i=j, \\ D_{i j}, & j=1, i \neq j, \\ \mathbf{0}_{3 \times 3}, & \text {Otherwise},\end{cases}\\
B&=\mathsf{blkdiag}\left(B_1, \ldots, B_{2000}\right)\!.
\end{align*}}
The regions of interest are provided as follows for all \(i \in \{1, \ldots, 2000\}\): the state set is \(X_i = [-20, 20]^3\), the initial set is \(X_{0_i} = [-5, 5]^3\), and the unsafe set is
	\(
	X_{u_i} = [-20, -10] \times [-20, -15] \times [6.5, 20] \;\cup\; [10, 20] \times [5.5, 20] \times [6.5, 20].
	\)

\subsection{Duffing Oscillator Network with Ring Interconnection Topology}
We study a Duffing oscillator network with a ring topology, where each subsystem dynamic is characterized as, for all \(i \in \{2, \ldots, 2000\}\),
\begin{align}\label{duffing-ring}
		\begin{split}
			& \dot{x}_{i_1}= x_{i_2} + \nu_{i_1},\\
			& \dot {x}_{i_2}= 2\,x_{i_1}-0.5\, x_{i_2}-0.01\,x^3_{i_1} + 0.1\,{x}_{(i-1)_1} +\nu_{i_2},
		\end{split}
\end{align}
with the first subsystem being affected by the last one. One can rewrite the dynamics in~\eqref{duffing-ring} in the form of \eqref{sys2} with the actual $\mathcal{M}_i(x_i) = [x_{i_1}; x_{i_2};  x_{i_1}^3]$ and
\begin{align*}
		A_i &= \begin{bmatrix} 0 & 1 & 0  \\ 2 & -0.5 & -0.01  \end{bmatrix}\!\!, \:
		B_i = \mathds{I}_2,\:
		D_{ij} =  \begin{bmatrix} 0 & 0  \\ 0.1 &  0  \end{bmatrix}\!\!.
			\end{align*}
The dictionary of monomials up to the known degree of $3$ is constructed as 
\begin{align}\label{Dic1}
	\mathcal{M}_i(x_i)=[x_{i_1}; x_{i_2}; x_{i_1}x_{i_2}; x^2_{i_1}; x^2_{i_2}; x^2_{i_1}x_{i_2}; x_{i_1}x^2_{i_2}; x^3_{i_1}; x^3_{i_2}].
\end{align}
The matrices of the interconnected network can be arranged as
\begin{align*}
	A(x)~\text{as in ~\eqref{A(x)-duffing-ring}}, ~~B= \mathsf{blkdiag}(B_1,\ldots,B_{2000}).
	\end{align*}
In this benchmark, the regions of interest are given as follows for all \(i \in \{1, \ldots, 2000\}\): the state set is \(X_i = [-10, 10]^2\), the initial set is \(X_{0_i} = [-4, 4]^2\), and the unsafe set is
	\(
	X_{u_i} = [-10, -6] \times [-10, -5] \;\cup\; [6, 10] \times [5, 10].
	\) Simulation results for this case study are illustrated in Figure~\ref{fig:3}.

\begin{figure}
	\centering
	\begin{subfigure}[t]{0.24\textwidth}
		\includegraphics[width=\linewidth]{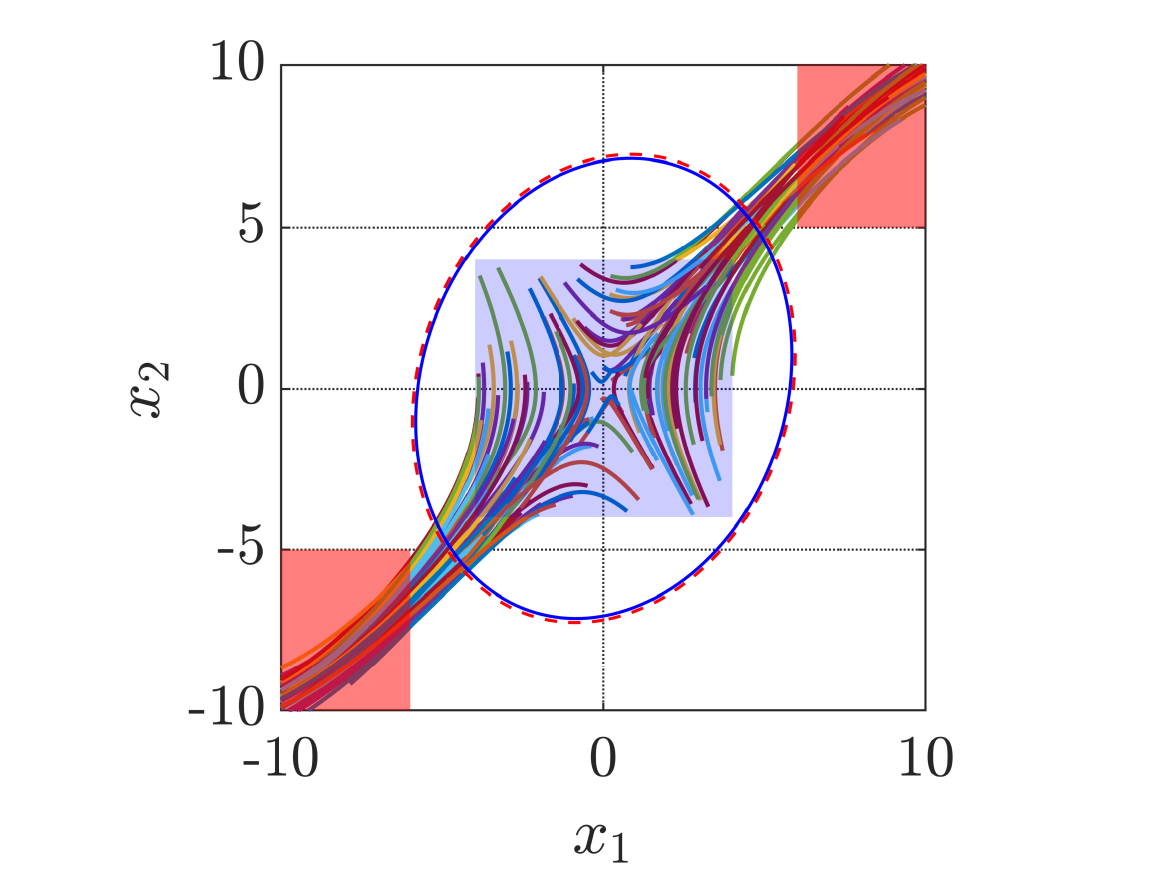}
		\caption{Open-loop trajectories}
	\end{subfigure}
	\hspace{-0.4cm}
	\begin{subfigure}[t]{0.24\textwidth}
		\includegraphics[width=\linewidth]{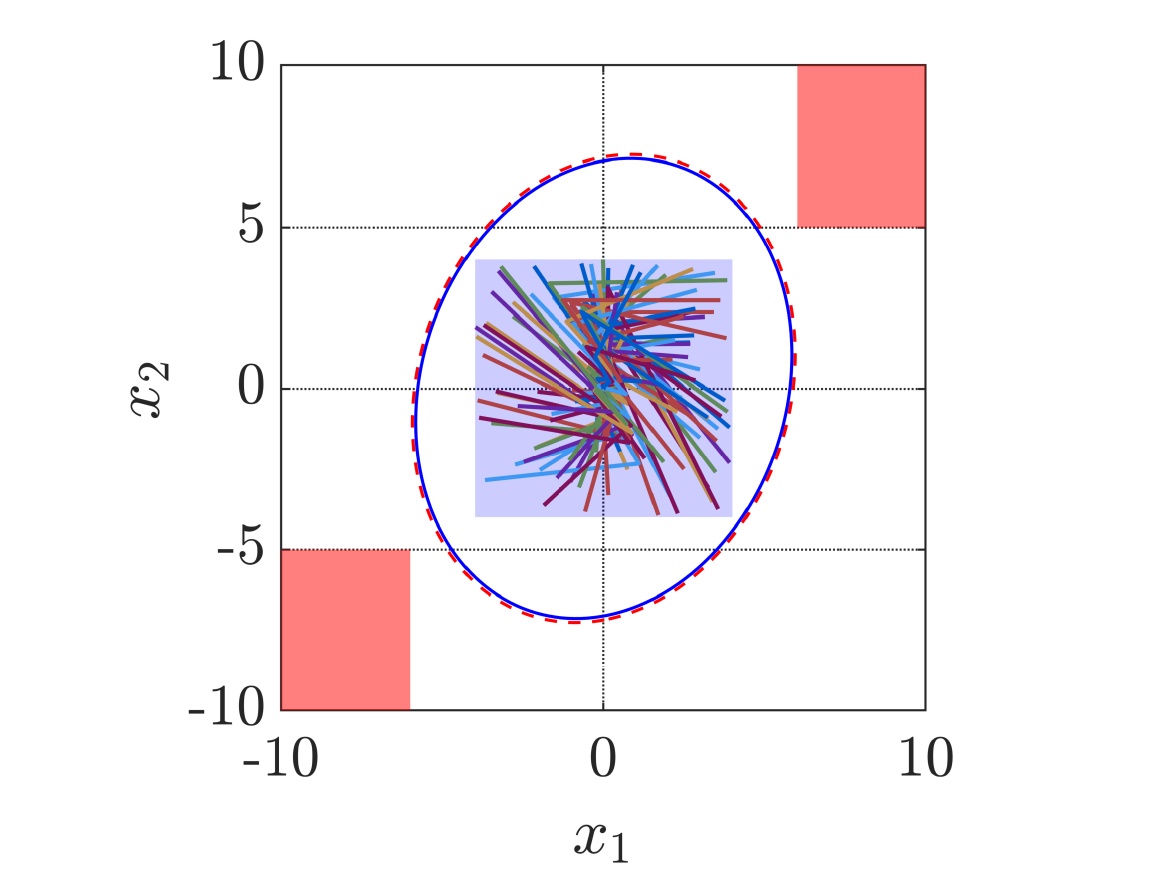}
		\caption{Closed-loop trajectories}
	\end{subfigure}
	\hspace{-0.4cm}
	\begin{subfigure}[t]{0.3\textwidth}
		\includegraphics[width=\linewidth]{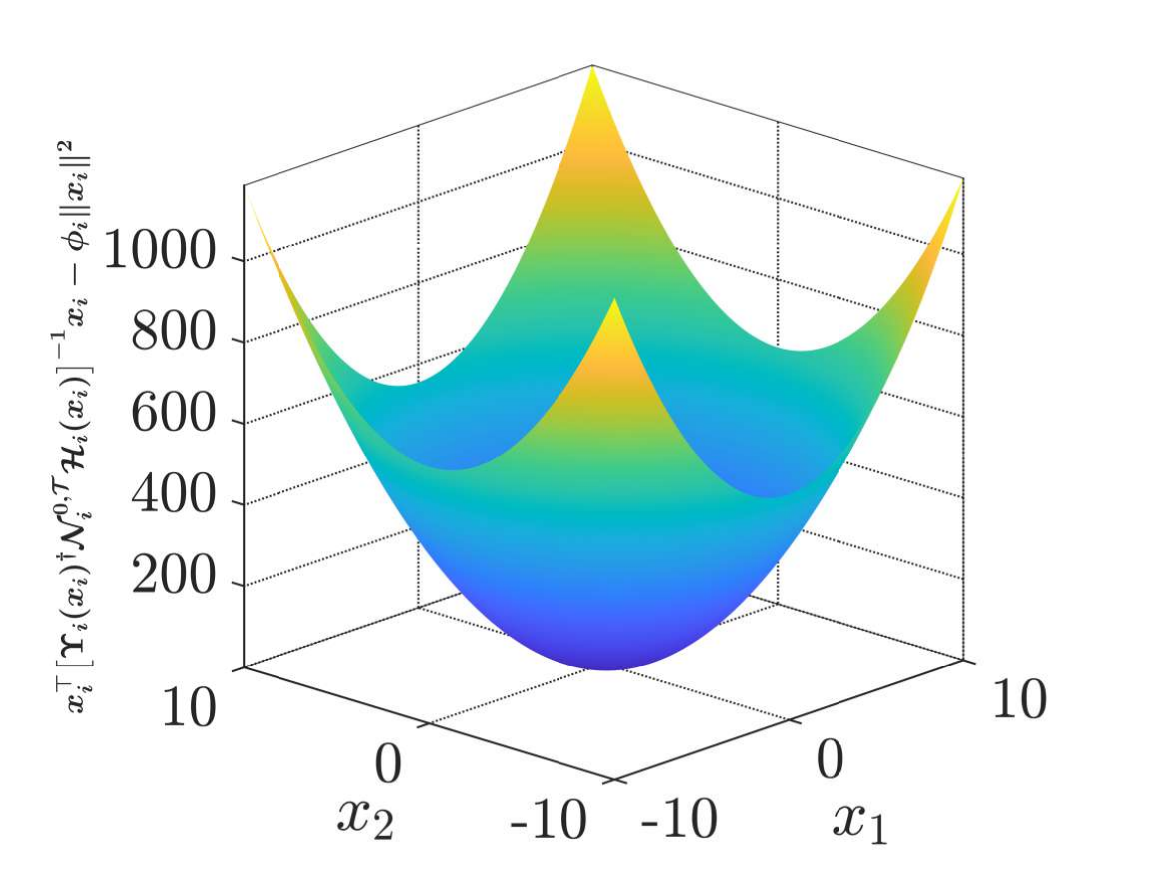}
		\caption{Verification of~\eqref{con}}
	\end{subfigure}
	
	\caption{Trajectories of 120 representative Duffing oscillators in a network of 2000 components with a ring interconnection topology. Plots \textbf{(a)} and \textbf{(b)} show the trajectories without and with the safety controller, respectively. The blue~\protect\bluesquare\ and red~\protect\redsquare\ boxes indicate the initial and unsafe regions, respectively, while the level sets \(\mathds{B}_i(x_i) = \gamma_i\) and \(\mathds{B}_i(x_i) = \beta_i\) are shown by~\sampleline{blue, thick} and~\sampleline{dashed, red, thick}, respectively. Plot \textbf{(c)} illustrates the satisfaction of condition~\eqref{con}, where the expression
			\(
			x_i^\top \left[ \Upsilon_i(x_i)^{\dagger} \mathcal{N}_i^{0, \mathcal{T}} \mathcal{H}_i(x_i) \right]^{-1} x_i - \phi_i \|x_i\|^2
			\)
			remains positive throughout the state space \([-10, 10]^2\).}
	\label{fig:3}
\end{figure}

\subsection{Duffing Oscillator Network with Binary Interconnection Topology}
We consider a Duffing oscillator network with a \emph{binary} interconnection topology, with each subsystem dynamic described as
\begin{align}\label{duffing-binary}
		\begin{split}
			\dot{x}_{i_1}&= x_{i_2} + \nu_{i_1},\\
			\dot {x}_{i_2}\!&=\! 2\,x_{i_1}\!-0.5\, x_{i_2}\!-0.01\,x^3_{i_1} \\ &\quad+ \!0.05\, \Omega\,(i=2 j \vee i=2 j+1) x_{j_1} +\! \nu_{i_2}, 
		\end{split}
\end{align}
where \(\vee\) denotes the logical \(\mathsf{OR}\) operation, and \(\Omega\) is an indicator defined as
\begin{equation*}
	\Omega\,(i=2 j \vee i=2 j+1)= \begin{cases}1, & i=2 j \vee i=2 j+1, \\ 0, & \text { Otherwise}.\end{cases}
\end{equation*}
One can rewrite the dynamics in~\eqref{duffing-binary} in the form of \eqref{sys2} with actual $\mathcal{M}_i(x_i) = [x_{i_1}; x_{i_2};  x_{i_1}^3]$ and
\begin{align*}
		A_i &= \begin{bmatrix} 0 & 1 & 0  \\ 2 & -0.5 & -0.01  \end{bmatrix}\!\!, \:
		B_i = \mathds{I}_2,\:
		D_{ij} =  \begin{bmatrix} 0 & 0  \\ 0.05 &  0  \end{bmatrix}\!\!.
\end{align*}
Furthermore, the dictionary of monomials up to the known degree of $3$ is constructed as in \eqref{Dic1}. The matrices of the interconnected network can be described as
\begin{align*}
		A(x)&=\left\{\mathsf{a}_{ij}\right\} = \begin{cases} A_i\Upsilon_{i}(x_{i}), & i=j, \\ D_{i j}, & i=2 j \vee i=2 j+1, \\ \mathbf{0}_{2 \times 2}, & \text{Otherwise}, \end{cases}\\
		B &= \mathsf{blkdiag}(B_1,\ldots,B_{1023}).
\end{align*}
The regions of interest are defined as follows for all \(i \in \{1, \ldots, 1023\}\): \(X_i = [-10, 10]^2\), \(X_{0_i} = [-4, 4]^2\), and 
\(
X_{u_i} = [-10, -6] \times [-10, -5] \;\cup\; [6, 10] \times [5, 10].
\)

\subsection{Chen Network with Fully-Interconnected Topology}
We analyze a Chen network with a fully-interconnected topology, where each subsystem dynamic is described as, for all \(i \in \{1, \ldots, 1000\}\),
\begin{align}\notag
	& \dot{x}_{i_1}=35\, x_{i_2}-35\, x_{i_1} - 5 \times 10^{-5}\!\!\!\! \sum_{\substack{j=1, j \neq i}}^{1000} x_{j_1},\\\notag
	& \dot {x}_{i_2}=-7\, x_{i_1}+28\,x_{i_2}-x_{i_1} x_{i_3} -  5 \times 10^{-5}\!\!\!\! \sum_{\substack{j=1, j \neq i}}^{1000} x_{j_2}+\nu_{i_1}, \\\label{chen-fully}
	& \dot{x}_{i_3}=x_{i_1} x_{i_2}-3\, x_{i_3} - 5 \times 10^{-5}\!\!\!\! \sum_{\substack{j=1, j \neq i}}^{1000} x_{j_3}+\nu_{i_2}.
	\end{align}
	One can rewrite the dynamics in~\eqref{chen-fully} in the form of \eqref{sys2} with the actual $\mathcal{M}_i(x_i) = [x_{i_1}; x_{i_2}; x_{i_3}; x_{i_1}x_{i_3}; x_{i_1}x_{i_2}]$, and
	\begin{align*}
	A_i \!\! = \!\!\! \begin{bmatrix} -35& 35 & 0 & 0 & 0  \\-7 & 28 & 0 & -1 & 0  \\ 0 & 0 & -3 & 0 & 1 \end{bmatrix}\!\!,  B_i \!\! = \!\!\! \begin{bmatrix} 0  & 0 \\ 1 & 0\\ 0 &1  \end{bmatrix}\!\!, D_{ij} \!\! = \!\! -5 \!\! \times \!\! 10^{-5} \mathds{I}_{3}.
\end{align*}
Additionally, the dictionary of monomials up to the known degree of $2$ is given by \eqref{Dic}. The matrices of the interconnected network can be presented as
\begin{align*}
	A(x)~\text{as in ~\eqref{A(x)-chen}}, ~~B= \mathsf{blkdiag}(B_1,\ldots,B_{2000}).
\end{align*}
The regions of interest are provided as follows for all \(i \in \{1, \ldots, 1000\}\): \(X_i = [-20, 20]^3\), \(X_{0_i} = [-2.5, 2.5]^3\), and 
	\(
	X_{u_i} = [-20, -4] \times [-20, -5] \times [-20, -4] \;\cup\; [3.5, 20] \times [5, 20] \times [4, 20].
	\) Figure~\ref{fig:4} depicts the trajectories of representative Chen and Lu subsystems with fully connected and star interconnection topologies, respectively.
 
 \begin{table*}[t!]
 	\centering
 	\caption{Data-driven results for the \emph{heterogeneous} network. Parameters include: number of subsystems $\mathcal{Q}$; number of samples $\mathcal{T}$; number of state variables per subsystem $n_i$; noise bound coefficient $\bar{\varkappa}_i$ as in~\eqref{noise-bound-2}; initial and unsafe level sets $\gamma_i$ and $\beta_i$; interaction gain $\rho_i$; decay rate $\varepsilon_i$; runtime in seconds ($\mathsf{RT}$); and memory usage in megabits ($\mathsf{MU}$). The reported runtime and memory usage correspond to solving the SOS problem in Algorithm~\ref{Alg:1}.}
 	\label{tab:system-configurations-a}
 	\resizebox{\linewidth}{!}{\begin{tabular}{@{}llc|cccccccccc@{}}
 				\toprule
 				Type of network & Topology & $\mathcal{Q}$ & $\Theta_i$ & $\mathcal{T}$&  $n_i$ & $\bar{\varkappa}_i$ & $\gamma_i$ & $\beta_i$ & $\rho_i$ & $\varepsilon_i$ & \(\mathsf{RT}\) & \(\mathsf{MU}\) \\ 
 				\midrule
 				\midrule
 				\multirow{5}{*}{Heterogeneous network} && &  $ \Theta_1 -\Theta_{300}$& $20$&$2$  & $0.18$   &    $121.13$  & $ 123.37$  &   $1.15 \times 10^{-6}$ & $0.99$&$3.06$
 				& $1.75$ \\
 				\cline{4-13}
 				&& &  $\Theta_{301}$& $20$&$ 2$ & $0.18$  &   $125.17$  & $127.34$   & $0.001$  & $0.99$ & $3.08$ &$1.75$ \\
 				\cline{4-13}
 				&\quad Line& $900$& $\Theta_{302} - \Theta_{600}$ & $20$ & $2$ & $0.18$   &   $123.19$  &   $125.36$ &   $4.62  \times 10^{-6}$  & $0.99$& $3.08$
 				& $1.75$  \\
 				\cline{4-13}
 				&& &  $\Theta_{601}$& $20$& $2$ &  $0.18$ &  $  127.88$  &   $ 129.92$  &  $ 0.001$  & $0.99$& $3.05$ &  $1.75$  \\
 				\cline{4-13}
 				&& &  $\Theta_{602} - \Theta_{900}$ & $20$ & $2$  &   $0.18$ & $125.69$ &   $127.74$ &    $2.89 \times 10^{-5}$&   $0.99$&  $3.09$&$1.75$
 				\\
 				\bottomrule
 	\end{tabular}}
 \end{table*}
 
  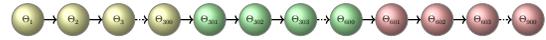
\begin{figure}[t!]
 	\centering
 	\resizebox{0.8\linewidth}{!}{%
 		\begin{tikzpicture}[
 			node distance=2cm,
 			vertex/.style={
 				circle,
 				draw=black!50,
 				thick,
 				minimum size=4em,     
 				inner sep=1pt,          
 				ball color=#1,          
 				shading=ball,           
 				font=\large       
 			},
 			edge/.style={->,ultra thick}
 			]
 			\node[vertex={yellow!30}] (n1) {$\Theta_{1}$};
 			
 			\foreach \i/\lab/\col in {
 				2/2/yellow!30,   3/3/yellow!30,   4/300/yellow!30,
 				5/301/green!30,  6/302/green!30,  7/303/green!30,  8/600/green!30,
 				9/601/red!30, 10/602/red!30, 11/603/red!30, 12/900/red!30
 			}{
 				\node[vertex={\col}] (n\i)
 				[right of=n\the\numexpr\i-1\relax] {$\Theta_{\lab}$};
 				\pgfmathtruncatemacro{\r}{mod(\i,4)}
 				\ifnum\r=0
 				\draw[edge,dash dot] (n\the\numexpr\i-1\relax) -- (n\i);
 				\else
 				\draw[edge]               (n\the\numexpr\i-1\relax) -- (n\i);
 				\fi
 			}
 		\end{tikzpicture}%
 	}
 	\caption{A heterogeneous network of $900$ subsystems arranged in a line topology, partitioned into three homogeneous subnetworks shown in yellow, green, and red. Each subnetwork consists of $300$ subsystems interconnected via a line topology.	}
 	\label{Het-fig}
 \end{figure}
 
 \begin{figure*}[!t] 
 \begin{align}\label{A(x)-duffing-ring}
 			A(x)  & = \begin{bmatrix}
 				A_1\Upsilon_1(x_1)  \!\! & \!\! \mathbf{0}_{2 \times 2} \!\! & \!\! \mathbf{0}_{2 \times 2} \!\! & \!\! \cdots \!\! & \!\! D_{1\,\,2000} \\
 				D_{2\,1} \!\! & \!\! A_2 \Upsilon_2(x_2) \!\! & \!\! \mathbf{0}_{2 \times 2} \!\! & \!\! \cdots \!\! & \!\! \mathbf{0}_{2 \times 2} \\
 				\vdots \!\! & \!\!  \!\! & \!\! \ddots \!\! & \!\!  \!\! & \!\! \vdots \\
 				\mathbf{0}_{2 \times 2} \!\! & \!\! \cdots \!\! & \!\! D_{1999\,\,1998} \!\! & \!\! A_{1999}\Upsilon_{1999}(x_{1999}) \!\! & \!\! \mathbf{0}_{2 \times 2} \\
 				\mathbf{0}_{2 \times 2} \!\! & \!\! \cdots \!\! & \!\! \mathbf{0}_{2 \times 2} \!\! & \!\! D_{2000\,\,1999} \!\! & \!\! A_{2000}\Upsilon_{2000}(x_{2000})
 			\end{bmatrix}
 	\end{align}
 \end{figure*}
 \begin{figure*}[!t] 
 \begin{align}
 			A(x)  & = \begin{bmatrix}\label{A(x)-chen}
 				A_1\Upsilon_1(x_1)  \!\! & \!\! D_{1\,2} \!\! & \!\! D_{1\,3} \!\! & \!\! \ldots \!\! & \!\! D_{1\,\,1000} \\
 				D_{2\,1} \!\! & \!\! A_2\Upsilon_2(x_2)  \!\! & \!\! D_{2\,3} \!\! & \!\! \ldots \!\! & \!\! D_{2\,\,1000} \\
 				\vdots \!\! & \!\! \quad \!\! & \!\! \ddots \!\! & \!\! \quad \!\! & \!\! \vdots \\
 				D_{999\,\,1} \!\! & \!\! \ldots \!\! & \!\! D_{999\,\,998} \!\! & \!\! A_{999}\Upsilon_{999}(x_{999})  \!\! & \!\! D_{999\,\,1000} \\
 				D_{1000\,\,1} \!\! & \!\! \ldots \!\! & \!\! D_{1000\,\,998} \!\! & \!\! D_{1000\,\,999} \!\! & \!\! A_{1000}\Upsilon_{1000}(x_{1000}) 
 			\end{bmatrix}
 	\end{align}
 \end{figure*}
 \begin{figure*}[!t] 
 \begin{align}\label{Het-net}
 			A(x) &=  \begin{bmatrix}
 				A_1\Upsilon_1(x_1)                 \!\! & \!\!\mathbf{0}_{2 \times 2} \!\! & \!\!\mathbf{0}_{2 \times 2} \!\! & \!\!\cdots \!\! & \!\!\mathbf{0}_{2 \times 2}  \\
 				D_{2\,1}     \!\! & \!\!A_2  \Upsilon_2(x_2)                   \!\! & \!\!\mathbf{0}_{2 \times 2} \!\! & \!\!\cdots \!\! & \!\!\mathbf{0}_{2 \times 2} \\
 				\!\! & \!\!                       \!\! & \!\!                       \!\! & \!\!       \!\! & \!\!                   \\
 				\vdots             \!\! & \!\!                       \!\! & \!\!\ddots                 \!\! & \!\!       \!\! & \!\!\vdots           \\
 				\!\! & \!\!                       \!\! & \!\!                       \!\! & \!\!       \!\! & \!\!                   \\
 				\mathbf{0}_{2 \times 2} \!\! & \!\!\cdots       \!\! & \!\!D_{899\,\,898}     \!\! & \!\!A_{899}\Upsilon_{899}(x_{899})  \!\! & \!\!\mathbf{0}_{2 \times 2} \\
 				\mathbf{0}_{2 \times 2} \!\! & \!\!\cdots       \!\! & \!\!\mathbf{0}_{2 \times 2}  \!\! & \!\!D_{900\,\,899} \!\! & \!\!A_{900}\Upsilon_{900}(x_{900}) 
 			\end{bmatrix}
 	\end{align}
 \end{figure*}
 
 \begin{figure}[t!]
 	\centering
 	\begin{subfigure}[t]{0.24\textwidth}
 		\includegraphics[width=\linewidth]{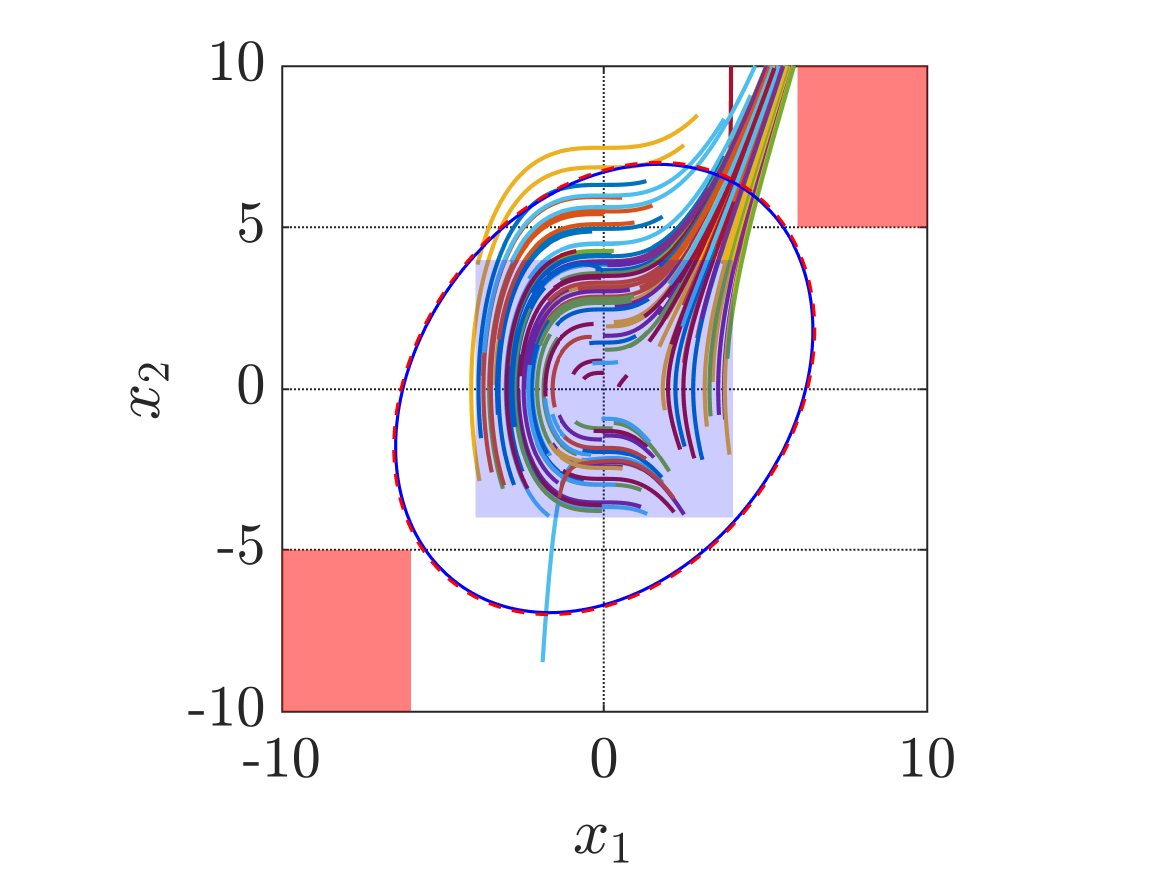}
 		\caption{Open-loop trajectories}
 	\end{subfigure}
 	\hspace{-0.4cm}
 	\begin{subfigure}[t]{0.24\textwidth}
 		\includegraphics[width=\linewidth]{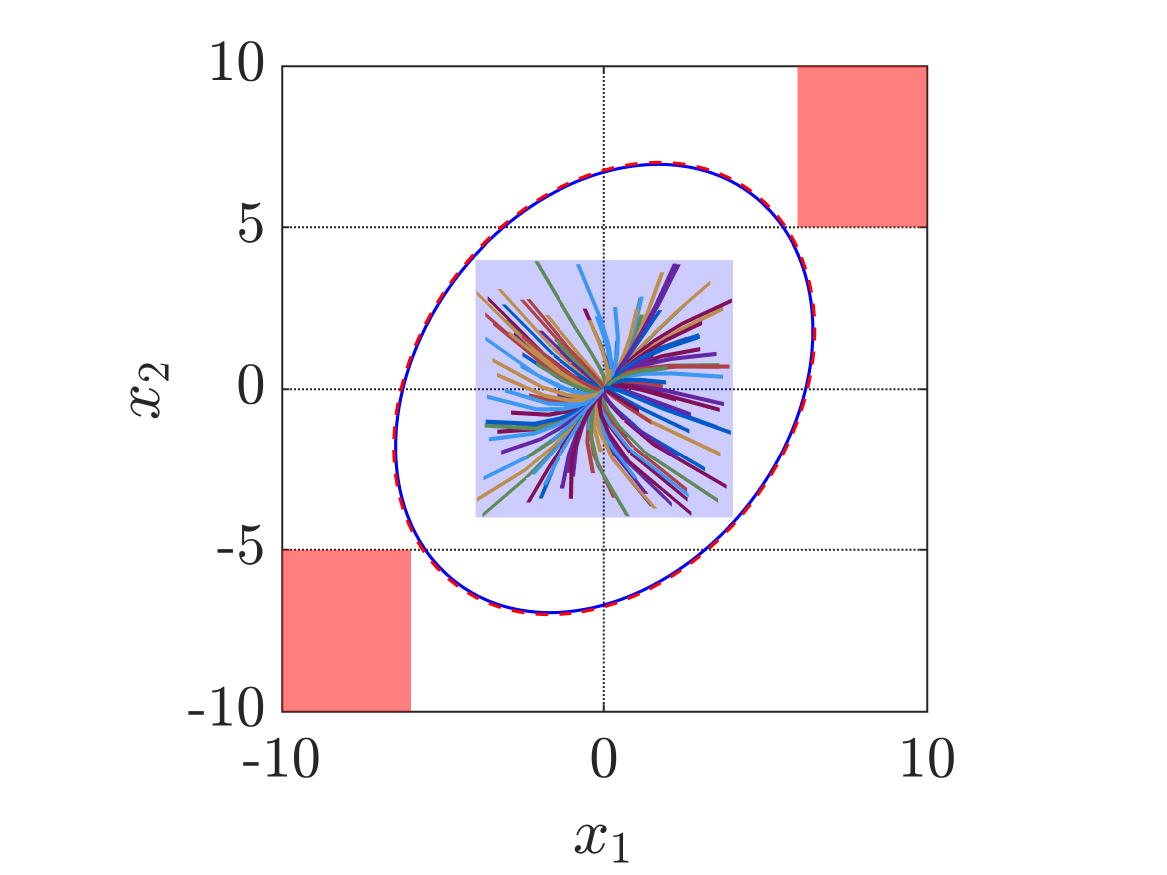}
 		\caption{Closed-loop trajectories}
 	\end{subfigure}
 	\hspace{-0.4cm}
 	\begin{subfigure}[t]{0.3\textwidth}
 		\includegraphics[width=\linewidth]{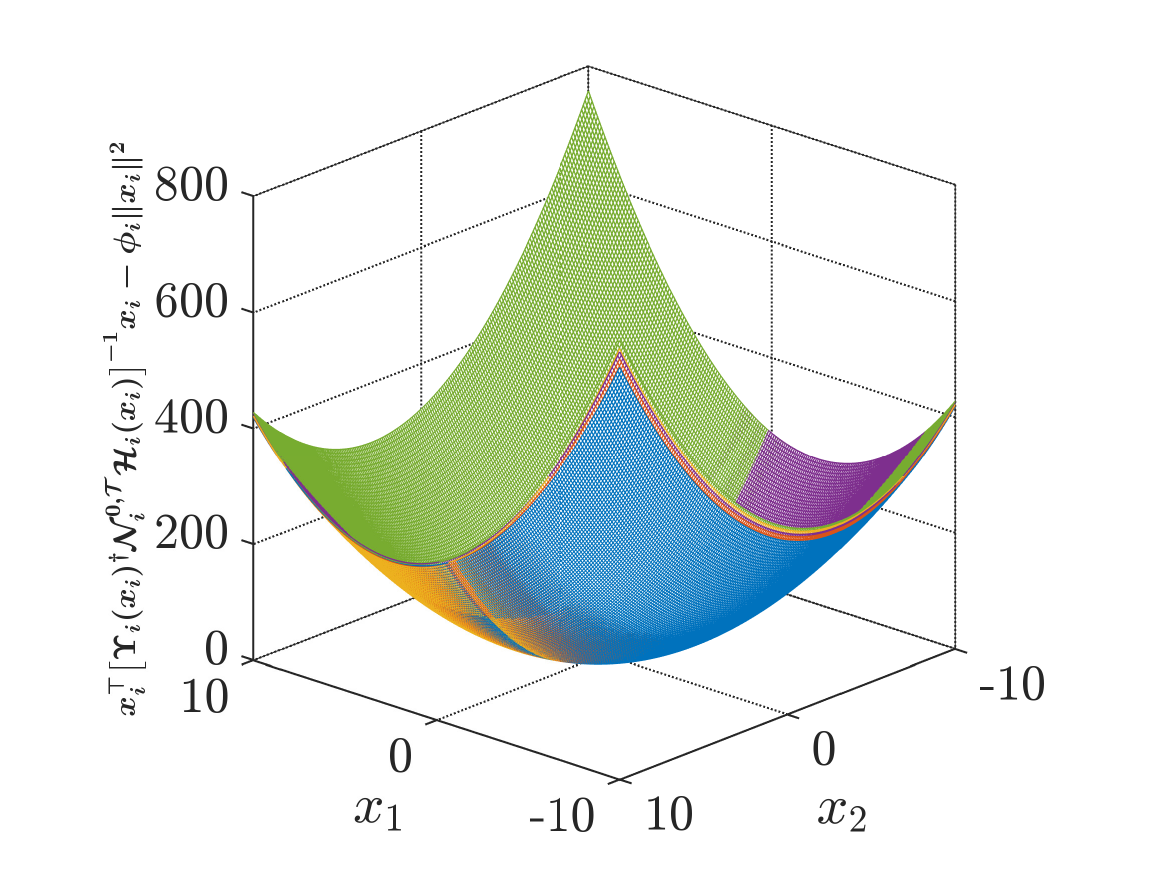}
 		\caption{Verification of~\eqref{con}}
 	\end{subfigure}
 	\caption{Trajectories of 120 representative subsystems in a heterogeneous network comprising 900 subsystems with a line interconnection topology. Plots \textbf{(a)} and \textbf{(b)} show the trajectories without and with the safety controller, respectively. The blue~\protect\bluesquare\ and red~\protect\redsquare\ boxes indicate the corresponding initial and unsafe regions, while the level sets \(\mathds{B_i}(x_i) = \gamma_i\) and \(\mathds{B_i}(x_i) = \beta_i\) are shown by~\sampleline{blue, thick} and~\sampleline{dashed, red, thick}, respectively. Plot \textbf{(c)} illustrates the satisfaction of condition~\eqref{con}, where the expression
 			\(
 			x_i^\top \left[ \Upsilon_i(x_i)^{\dagger} \mathcal{N}_i^{0,\mathcal{T}} \mathcal{H}_i(x_i) \right]^{-1} x_i - \phi_i \|x_i\|^2
 			\)
 			remains positive throughout the state space \([-10, 10]^2\) for five different subsystems, including subsystems 301 and 601 (cf. Figure~\ref{Het-fig}).}
 	\label{fig:5}
 \end{figure}
 
 \subsection{Heterogeneous Network with Line Interconnection Topology}
 {We consider a \emph{heterogeneous} network obtained by serially connecting three homogeneous subnetworks through a line topology (cf. Figure~\ref{Het-fig}). Each subsystem comprises 300 subsystems, yielding 900 subsystems in total. For all \(i \in \{2, \ldots, 300\}\), the subsystem dynamics involve two state variables \(x_i = [x_{i_1};\, x_{i_2}]\), and are characterized by
\begin{align}
		\begin{split}\label{het}
			& \dot{x}_{i_1}= x_{i_2} - 10^{-3}\,{x}_{(i-1)_1} + \nu_{i_1},\\
			& \dot{x}_{i_2}= x_{i_1}^2 - 10^{-3}\,{x}_{(i-1)_2} + \nu_{i_2}, 
		\end{split}
\end{align}
where the first subsystem does not receive any internal inputs. One can rewrite \eqref{het} in the form of~\eqref{sys2} with the actual \(\mathcal{M}_i(x_i) = [x_{i_1};x_{i_2};\,x^2_{i_2}]\), and
\begin{align*}
		A_i &= \begin{bmatrix} 0 &1 & 0 \\ 0 & 0 & 1 \end{bmatrix}, \:
		B_i = \mathds{I}_2, \:
		D_{ij} = -10^{-3} \mathds{I}_2.
\end{align*}
For \(i=301\), the dynamics are
\begin{align*}
		& \dot{x}_{i_1}= 1.2\,x_{i_2} - 3\times10^{-2}\,{x}_{(i-1)_1} + \nu_{i_1},\\
		& \dot{x}_{i_2}= 1.2\,x_{i_1}^2 - 3\times10^{-2}\,{x}_{(i-1)_2} + \nu_{i_2},
\end{align*}
which can be rewritten in the form of~\eqref{sys2} with \(\mathcal{M}_i(x_i) = [x_{i_1};x_{i_2};\,x^2_{i_2}]\), alongside
\begin{align*}
		A_i &= \begin{bmatrix} 0 & 1.2 & 0 \\ 0& 0 & 1.2 \end{bmatrix}\!\!, \:
		B_i = \mathds{I}_2, \:
		D_{ij} = -3\times10^{-2} \mathds{I}_2.
\end{align*}
For all \(i \in \{302,\ldots,600\}\), the dynamics are
\begin{align*}
		& \dot{x}_{i_1}= 1.2\,x_{i_2} - 2\times10^{-3}\,{x}_{(i-1)_1} + \nu_{i_1},\\\notag
		& \dot{x}_{i_2}= 1.2\,x_{i_1}^2 - 2\times10^{-3}\,{x}_{(i-1)_2} + \nu_{i_2},
\end{align*}
which can be expressed in the form of~\eqref{sys2} using \(\mathcal{M}_i(x_i) = [x_{i_1};x_{i_2};\,x^2_{i_2}]\), along with
\begin{align*}
		A_i &= \begin{bmatrix} 0 & 1.2 & 0 \\ 0 & 0 & 1.2 \end{bmatrix}\!\!, \:
		B_i = \mathds{I}_2, \:
		D_{ij} = -2\times10^{-3}  \mathds{I}_2.
\end{align*}
For \(i=601\), the subsystem’s dynamics are
\begin{align*}
		& \dot{x}_{i_1}= 1.4\,x_{i_2} - 4\times10^{-2}\,{x}_{(i-1)_1} + \nu_{i_1},\\
		& \dot{x}_{i_2}= 1.4\,x_{i_1}^2 - 4\times10^{-2}\,{x}_{(i-1)_2} + \nu_{i_2},
\end{align*}
where the form of~\eqref{sys2} can be obtained by \(\mathcal{M}_i(x_i) = [x_{i_1};x_{i_2};\,x^2_{i_2}]\), and
\begin{align*}
		A_i &= \begin{bmatrix} 0 & 1.4 & 0 \\0 &  0 & 1.4 \end{bmatrix}\!\!, \:
		B_i = \mathds{I}_2, \:
		D_{ij} = -4\times10^{-2}  \mathds{I}_2.
\end{align*}
For all \(i \in \{602,\ldots,900\}\), the dynamics are
\begin{align}\notag
		& \dot{x}_{i_1}= 1.4\,x_{i_2} - 5\times10^{-3}\,{x}_{(i-1)_1} + \nu_{i_1},\\\notag
		& \dot{x}_{i_2}= 1.4\,x_{i_1}^2 - 5\times10^{-3}\,{x}_{(i-1)_2} + \nu_{i_2},
\end{align}
being in the form of~\eqref{sys2} with \(\mathcal{M}_i(x_i) = [x_{i_1};x_{i_2};\,x^2_{i_2}]\), and
\begin{align*}
		A_i &= \begin{bmatrix} 0 & 1.4 & 0 \\ 0 &0 & 1.4 \end{bmatrix}\!\!, \:
		B_i = \mathds{I}_2, \:
		D_{ij} = -5\times10^{-3}  \mathds{I}_2.
\end{align*}
The matrices \(A_i\) and \(B_i\), along with the exact form of \(\mathcal{M}_i(x_i)\), are unknown. However, the dictionary of monomials up to the known degree of $2
$ is defined as \(\mathcal{M}_i(x_i) = [x_{i_1}; x_{i_2}; x_{i_1}x_{i_2}; x^2_{i_1}; x^2_{i_2}],\) for all \(i \in \{1,\ldots,900\}\). The matrices of the interconnected network can be arranged as
\begin{align*}
	A(x)~\text{as in ~\eqref{Het-net}}, ~~B= \mathsf{blkdiag}(B_1,\ldots,B_{900}).
	\end{align*}} 
The regions of interest are given as follows for all \(i \in \{1, \ldots, 900\}\): \(X_i = [-10, 10]^2\), \(X_{0_i} = [-4, 4]^2\), and
	\(
	X_{u_i} = [-10, -6] \times [-10, -5] \;\cup\; [6, 10] \times [5, 10].
	\) The designed parameters \(\bar{\varkappa}_i\), \(\gamma_i\), \(\beta_i\), \(\rho_i\), and \(\varepsilon_i\) are provided in Table~\ref{tab:system-configurations-a}. The simulation results are also presented in Figure~\ref{fig:5}.

\section{Conclusion}\label{sec:Conclusion}
We developed a compositional data-driven scheme with noisy data to design robust safety controllers for large-scale interconnected networks with unknown dynamics. Our approach addressed the dual challenges of high network dimensionality and unknown model complexity by considering the network as a collection of smaller subsystems and utilizing data from each subsystem's trajectory to design control sub-barrier certificates and their corresponding local controllers. Our method requires only a single noise-corrupted input-state trajectory from each subsystem up to a specified time horizon, provided a certain rank condition is met. By leveraging the small-gain compositional condition, we then composed CSBC derived from data into a control barrier certificate for the entire network, ensuring safety over an infinite time horizon with correctness guarantees. Extending our compositional data-driven approach to accommodate more complex dynamical systems beyond polynomial is an active direction for future work.

\bibliographystyle{IEEEtran}
\bibliography{biblio}	

\appendix
Here, we present a detailed presentation of the CSBC \(\mathds{B}_i(x_i)\), local controllers \(\nu_i\), the CBC \(\mathds{B}(x)\), the safety controller \(\nu\), and the network parameters \(\gamma\), \(\beta\), and \(\varepsilon\), all designed from data.

\subsection{Lorenz Network with Fully-Interconnected  Topology}
\textbf{Subsystems: CSBC and local controllers}
\begin{align*}
		\nu_{i_1} &= 
		21.8533\,x_{i_1}^2 - 67.7554\,x_{i_1}x_{i_2} - 67.6251\,x_{i_1}x_{i_3} \\
		&\quad + 47.7192\,x_{i_2}^2 + 96.3149\,x_{i_2}x_{i_3} + 55.7623\,x_{i_3}^2 \\
		&\quad - 1237.1117\,x_{i_1} + 197.717\,x_{i_2} + 260.4928\,x_{i_3} \\[1em]
		\nu_{i_2} &= 
		117.0317\,x_{i_1}^2 - 138.0305\,x_{i_1}x_{i_2} - 164.201\,x_{i_1}x_{i_3} \\
		&\quad + 6.5751\,x_{i_2}^2 + 12.0967\,x_{i_2}x_{i_3} + 12.9936\,x_{i_3}^2 \\
		&\quad + 51.4842\,x_{i_1} - 318.5331\,x_{i_2} - 109.3699\,x_{i_3} \\[1em]
		\nu_{i_3} &= 
		19.0949\,x_{i_1}^2 - 59.5318\,x_{i_1}x_{i_2} - 105.0946\,x_{i_1}x_{i_3} \\
		&\quad + 33.7661\,x_{i_2}^2 + 53.6538\,x_{i_2}x_{i_3} + 27.7845\,x_{i_3}^2 \\
		&\quad + 314.9155\,x_{i_1} - 270.904\,x_{i_2} - 458.6184\,x_{i_3} \\[1em]
		\mathds{B}_i(x_i) &= 
		18.7668\,x_{i_1}^2 - 5.8139\,x_{i_1}x_{i_2} - 7.728\,x_{i_1}x_{i_3} \\
		&\quad + 5.5492\,x_{i_2}^2 + 6.0589\,x_{i_2}x_{i_3} + 6.3392\,x_{i_3}^2
\end{align*}
\textbf{Network: CBC and safety controller}
\begin{align*}
		\mathds{B}(x)&=\sum_{i=1}^{1000} \mathds{B}_i\left( x_i\right),~\nu = [\nu_1;\dots; \nu_{1000}]\\ \gamma &=\sum_{i=1}^{1000} \gamma_i = 4.79 \times 10^5,~ \beta=\sum_{i=1}^{1000} \beta_i= 4.9 \times 10^5\\\varepsilon &=-\varpi = 0.98
\end{align*}
\subsection{ Lorenz Network with Ring Interconnection Topology}
\textbf{Subsystems: CSBC and local controllers}
\begin{align*}
		\nu_{i_1} &= 
		-10.0712\,x_{i_1}^2 + 23.9253\,x_{i_1}x_{i_2} + 8.0498\,x_{i_1}x_{i_3} \\
		&\quad - 25.8062\,x_{i_2}^2 - 6.5798\,x_{i_2}x_{i_3} + 6\,x_{i_3}^2\\
			&\quad - 2734.5265\,x_{i_1} + 244.7904\,x_{i_2} + 249.8385\,x_{i_3} \\[1em]
			\nu_{i_2} &= 
		-32.134\,x_{i_1}^2 + 135.5226\,x_{i_1}x_{i_2} + 35.9192\,x_{i_1}x_{i_3} \\
		&\quad - 13.1452\,x_{i_2}^2 - 8.983\,x_{i_2}x_{i_3} - 42.2782\,x_{i_3}^2 \\
		&\quad + 512.586\,x_{i_1} - 536.8531\,x_{i_2} - 132.6197\,x_{i_3} \\[1em]
			\nu_{i_3} &= 
		-24.4745\,x_{i_1}^2 + 1.1856\,x_{i_1}x_{i_2} - 46.9177\,x_{i_1}x_{i_3} \\
		&\quad - 7.6293\,x_{i_2}^2 + 36.0765\,x_{i_2}x_{i_3} + 8.4757\,x_{i_3}^2 \\
		&\quad - 42.8253\,x_{i_1} - 53.8814\,x_{i_2} - 583.1774\,x_{i_3}
\end{align*}

\begin{align*}
	\mathds{B}_i(x_i) &= 
	26.7769\,x_{i_1}^2 - 5.821\,x_{i_1}x_{i_2} - 3.9914\,x_{i_1}x_{i_3} \\
	&\quad + 5.1485\,x_{i_2}^2 + 1.923\,x_{i_2}x_{i_3} + 6.1772\,x_{i_3}^2
\end{align*}

\textbf{Network: CBC and safety controller}
\begin{align*}
		\mathds{B}(x)&=\sum_{i=1}^{2000} \mathds{B}_i\left( x_i\right),~\nu = [\nu_1;\dots; \nu_{2000}]\\ \gamma &=\sum_{i=1}^{2000} \gamma_i = 10^6,~ \beta=\sum_{i=1}^{2000} \beta_i= 1.03 \times 10^6\\\varepsilon &=-\varpi = 0.98
\end{align*}
\subsection{Spacecraft Network with Line Interconnection Topology}
\textbf{Subsystems: CSBC and local controllers}
\begin{align*}
		\nu_{i_1} &= 
		-50.6439\,x_{i_1}^2 - 114.2867\,x_{i_1}x_{i_2} + 100.2839\,x_{i_1}x_{i_3} \\
		&\quad + 121.479\,x_{i_2}^2 - 244.9116\,x_{i_2}x_{i_3} + 140.0651\,x_{i_3}^2 \\
		&\quad - 18245.2415\,x_{i_1} - 1172.2093\,x_{i_2} - 6485.0942\,x_{i_3} \\[1em]
		\nu_{i_2} &= 
		-73.0867\,x_{i_1}^2 - 149.1322\,x_{i_1}x_{i_2} + 181.8506\,x_{i_1}x_{i_3} \\
		&\quad + 72.7091\,x_{i_2}^2 + 232.6915\,x_{i_2}x_{i_3} - 29.0203\,x_{i_3}^2 \\
		&\quad - 770.8096\,x_{i_1} - 8453.0782\,x_{i_2} - 3753.6522\,x_{i_3} \\[1em]
		\nu_{i_3} &= 
		116.7209\,x_{i_1}^2 + 546.4828\,x_{i_1}x_{i_2} - 430.1346\,x_{i_1}x_{i_3} \\
		&\quad - 352.0659\,x_{i_2}^2 + 137.7694\,x_{i_2}x_{i_3} - 144.2429\,x_{i_3}^2 \\
		&\quad - 10470.1736\,x_{i_1} - 5060.5\,x_{i_2} - 19797.8908\,x_{i_3} \\[1em]
		\mathds{B}_i(x_i) &= 
		4.886\,x_{i_1}^2 + 0.54151\,x_{i_1}x_{i_2} + 3.6833\,x_{i_1}x_{i_3} \\
		&\quad + 2.2014\,x_{i_2}^2 + 1.7819\,x_{i_2}x_{i_3} + 3.4979\,x_{i_3}^2
\end{align*}
\textbf{Network: CBC and safety controller}
\begin{align*}
		\mathds{B}(x)&=\sum_{i=1}^{2000} \mathds{B}_i\left( x_i\right),~\nu = [\nu_1;\dots; \nu_{2000}]\\ \gamma &=\sum_{i=1}^{2000} \gamma_i = 1.46 \times 10^5~ \beta=\sum_{i=1}^{2000} \beta_i= 1.5 \times 10^5\\\varepsilon &=-\varpi =0.05
\end{align*}
\subsection{Lu Network with Star Interconnection Topology}
\textbf{Subsystems: CSBC and local controllers}
\begin{align*}
		\nu_{i_1} &= 
		59.8425\,x_{i_1}^2 + 37.8192\,x_{i_1}x_{i_2} + 251.9851\,x_{i_1}x_{i_3} \\
		&\quad + 5.5285\,x_{i_2}^2 + 74.9351\,x_{i_2}x_{i_3} + 252.0663\,x_{i_3}^2 \\
		&\quad + 245.721\,x_{i_1} - 315.0832\,x_{i_2} + 853.58\,x_{i_3} \\[1em]
		\nu_{i_2} &= 
		-91.8658\,x_{i_1}^2 - 185.3774\,x_{i_1}x_{i_2} - 245.4829\,x_{i_1}x_{i_3} \\
		&\quad - 59.3153\,x_{i_2}^2 - 377.5074\,x_{i_2}x_{i_3} - 34.1371\,x_{i_3}^2 \\
		&\quad - 1117.3166\,x_{i_1} - 1543.8441\,x_{i_2} - 370.2911\,x_{i_3} \\[1em]
		\mathds{B}_i(x_i) &= 
		4.1535\,x_{i_1}^2 + 6.9638\,x_{i_1}x_{i_2} + 4.5311\,x_{i_1}x_{i_3} \\
		&\quad + 5.5811\,x_{i_2}^2 + 1.0341\,x_{i_2}x_{i_3} + 3.7818\,x_{i_3}^2
\end{align*}
\textbf{Network: CBC and safety controller}
\begin{align*}
		\mathds{B}(x)&=\sum_{i=1}^{2000} \mathds{B}_i\left( x_i\right),~\nu = [\nu_1;\dots; \nu_{2000}]\\ \gamma &=\sum_{i=1}^{2000} \gamma_i = 1.46 \times 10^6,~ \beta=\sum_{i=1}^{2000} \beta_i= 1.5 \times 10^6\\\varepsilon &=-\varpi =0.88
\end{align*}
\subsection{Duffing Oscillator Network with Ring Interconnection Topology}
\textbf{Subsystems: CSBC and local controllers}
	\begin{align*}
		\nu_{i_1} &= 
		-2.0346\,x_{i_1}^3 - 13.4461\,x_{i_1}^2x_{i_2} - 5.1204\,x_{i_1}x_{i_2}^2\\&\quad - 10.3586\,x_{i_2}^3  + 1.2659\,x_{i_1}^2 - 21.134\,x_{i_1}x_{i_2}\\ &\quad + 10.1984\,x_{i_2}^2  - 1168.2521\,x_{i_1} + 144.0447\,x_{i_2} \\[1em]
		\nu_{i_2} &= 
		24.9364\,x_{i_1}^3 - 6.0576\,x_{i_1}^2x_{i_2} + 19.0191\,x_{i_1}x_{i_2}^2\\ &\quad - 5.2088\,x_{i_2}^3  + 36.8267\,x_{i_1}^2 - 24.4002\,x_{i_1}x_{i_2}\\ &\quad + 2.5832\,x_{i_2}^2  + 133.4182\,x_{i_1} - 778.8311\,x_{i_2} \\[1em]
		\mathds{B}_i(x_i) &= 
		8.4503\,x_{i_1}^2 - 2.0223\,x_{i_1}x_{i_2} + 5.6554\,x_{i_2}^2
\end{align*}
\textbf{Network: CBC and safety controller}
\begin{align*}
		\mathds{B}(x)&=\sum_{i=1}^{2000} \mathds{B}_i\left( x_i\right),~\nu = [\nu_1;\dots; \nu_{2000}]\\ \gamma &=\sum_{i=1}^{2000} \gamma_i = 5.63 \times 10^5,~ \beta=\sum_{i=1}^{2000} \beta_i=5.83 \times 10^5\\\varepsilon &=-\varpi =0.87
\end{align*}
\subsection{Duffing Oscillator Network with Binary Interconnection Topology}
\textbf{Subsystems: CSBC and local controllers}
\begin{align*}
		\nu_{i_1} &= 
		-1.7313\,x_{i_1}^3 - 14.8234\,x_{i_1}^2x_{i_2} - 4.6827\,x_{i_1}x_{i_2}^2\\
		&\quad  - 12.1856\,x_{i_2}^3 + 1.0959\,x_{i_1}^2 - 18.9204\,x_{i_1}x_{i_2}\\
		&\quad  + 2.9079\,x_{i_2}^2 - 640.2218\,x_{i_1} + 91.5416\,x_{i_2}\\[1em]
			\nu_{i_2} &= 
		24.1086\,x_{i_1}^3 - 6.3893\,x_{i_1}^2x_{i_2} + 19.748\,x_{i_1}x_{i_2}^2\\
		&\quad - 5.6392\,x_{i_2}^3  + 31.9524\,x_{i_1}^2 - 13.2464\,x_{i_1}x_{i_2}\\
		&\quad + 1.8457\,x_{i_2}^2  + 66.1871\,x_{i_1} - 474.8328\,x_{i_2} 
\end{align*}

\begin{align*}
\mathds{B}_i(x_i) = 
	8.0675\,x_{i_1}^2 - 2.0414\,x_{i_1}x_{i_2} + 6.0436\,x_{i_2}^2
\end{align*}
\textbf{Network: CBC and safety controller}
\begin{align*}
		\mathds{B}(x)&=\sum_{i=1}^{1023} \mathds{B}_i\left( x_i\right),~\nu = [\nu_1;\dots; \nu_{1023}]\\ \gamma &=\sum_{i=1}^{1023} \gamma_i = 2.87 \times 10^5,~ \beta=\sum_{i=1}^{1023} \beta_i= 2.97 \times 10^5\\\varepsilon &=-\varpi =0.93
\end{align*}
\subsection{Chen Network with Fully-Interconnected Topology}
\textbf{Subsystems: CSBC and local controllers}
\begin{align*}
		\nu_{i_1} &= 
		131.852\,x_{i_1}^2 + 174.5084\,x_{i_1}x_{i_2} + 788.7689\,x_{i_1}x_{i_3} \\
		&\quad + 57.7206\,x_{i_2}^2 + 495.3951\,x_{i_2}x_{i_3} + 874.5197\,x_{i_3}^2 \\
		&\quad - 2112.1811\,x_{i_1} - 56.1417\,x_{i_2} + 1193.6544\,x_{i_3} \\[1em]
		\nu_{i_2} &= 
		-86.624\,x_{i_1}^2 - 104.6167\,x_{i_1}x_{i_2} - 467.8595\,x_{i_1}x_{i_3} \\
		&\quad - 29.3959\,x_{i_2}^2 - 236.3971\,x_{i_2}x_{i_3} - 350.9533\,x_{i_3}^2 \\
		&\quad - 612.0299\,x_{i_1} - 519.3022\,x_{i_2} - 1292.8814\,x_{i_3} \\[1em]
		\mathds{B}_i(x_i) &= 
		18.9041\,x_{i_1}^2 + 11.2502\,x_{i_1}x_{i_2} + 17.1341\,x_{i_1}x_{i_3} \\
		&\quad + 2.5725\,x_{i_2}^2 + 10.1225\,x_{i_2}x_{i_3} + 12.6188\,x_{i_3}^2
\end{align*}
\textbf{Network: CBC and safety controller}
\begin{align*}
		\mathds{B}(x)&=\sum_{i=1}^{1000} \mathds{B}_i\left( x_i\right),~\nu = [\nu_1;\dots; \nu_{1000}]\\ \gamma &=\sum_{i=1}^{1000} \gamma_i = 5.14 \times 10^5,~ \beta=\sum_{i=1}^{1000} \beta_i= 5.27 \times 10^5\\\varepsilon &=-\varpi =0.38
\end{align*}
\subsection{Heterogeneous Network with Line Interconnection Topology}
Due to space constraints, details of the subsystems’ CSBC and local controllers are omitted and can be accessed by running the benchmark implementation code provided earlier.\\
\textbf{Network: CBC and safety controller.}
\begin{align*}
		\mathds{B}(x)&=\sum_{i=1}^{300} \left(3.0597\,x_{i_1}^2 - 1.4599\,x_{i_1}x_{i_2}+ 2.7439\,x_{i_2}^2\right)\\ & ~~~+ 3.1924\,x_{301_1}^2 - 1.5561\,x_{301_1}x_{301_2}+ 2.786\,x_{301_2}^2\\ & ~~~+\sum_{i=302}^{600} \left(3.1479\,x_{i_1}^2 - 1.5244\,x_{i_1}x_{i_2} + 2.7371\,x_{i_2}^2\right)\\
        & ~~~+ 3.3016\,x_{601_1}^2 -1.6452\,x_{601_1}x_{601_2}+ 2.7846\,x_{601_2}^2\\ & ~~~+\sum_{i=602}^{900} \left(3.2505\,x_{i_1}^2 - 1.6046\,x_{i_1}x_{i_2} + 2.7344\,x_{i_2}^2\right)\\
		\nu& = [\nu_1;\dots; \nu_{900}], \gamma =\sum_{i=1}^{900} \gamma_i = 1.11 \times 10^{5},\\ \beta&=\sum_{i=1}^{900} \beta_i= 1.13 \times 10^{5}, \varepsilon =-\varpi = 0.98
\end{align*}

\begin{IEEEbiography}[{\includegraphics[width=1in,height=1.25in,clip,keepaspectratio]{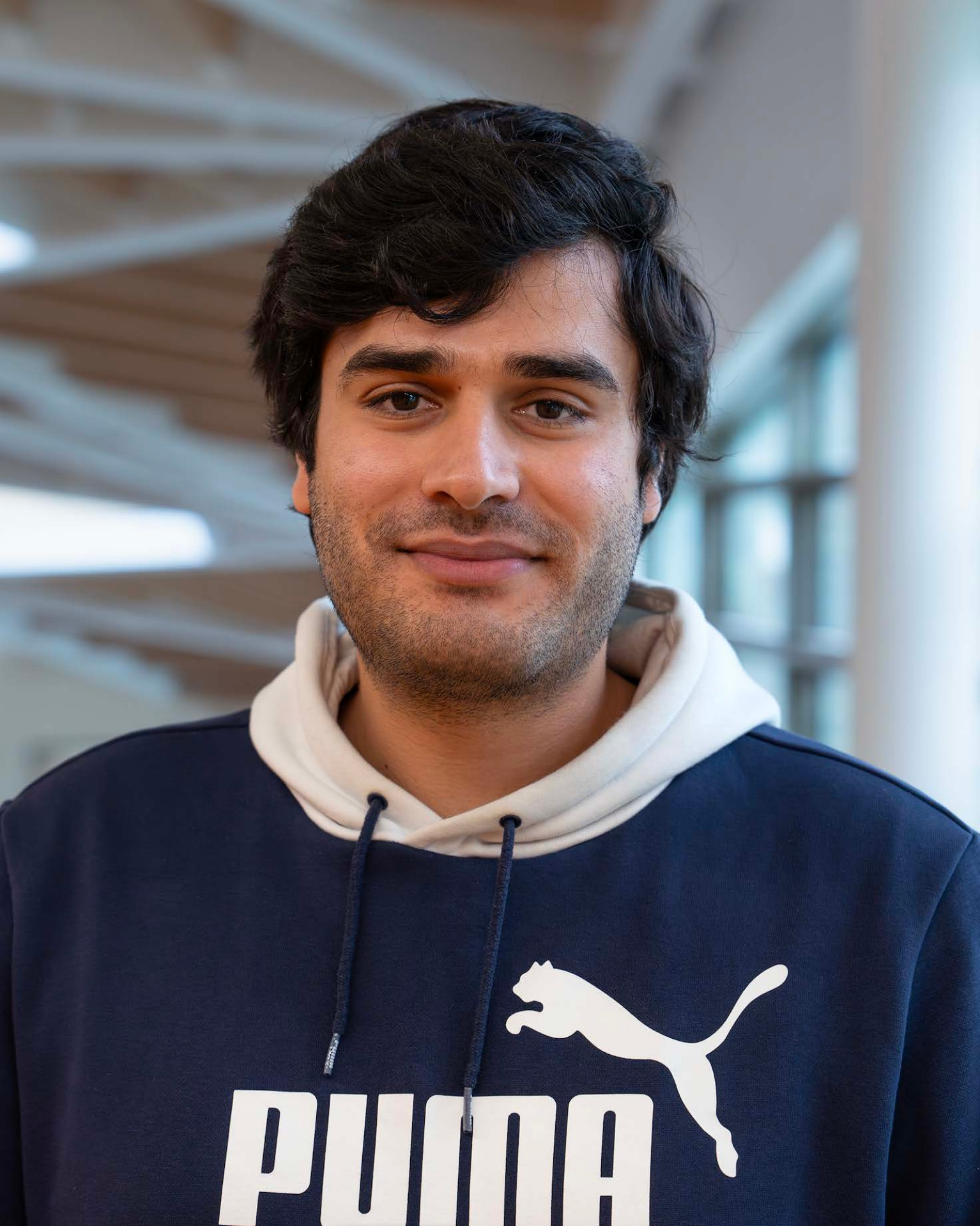}}]{Omid Akbarzadeh}~(Student Member, IEEE) is currently a PhD student in the School of Computing at Newcastle University, U.K. His academic journey commenced at Shiraz University, where he obtained a Bachelor of Science in Electrical and Electronic Engineering. Following this, he pursued a master's degree in Communications and Computer Network Engineering (CCNE) at the Polytechnic University of Turin, Italy (Politecnico di Torino). His research interests include safe cyber-physical systems, communication networks, data-driven approaches, and formal control.
\end{IEEEbiography}

\begin{IEEEbiography}[{\includegraphics[width=1in,height=1.3in,clip,keepaspectratio]{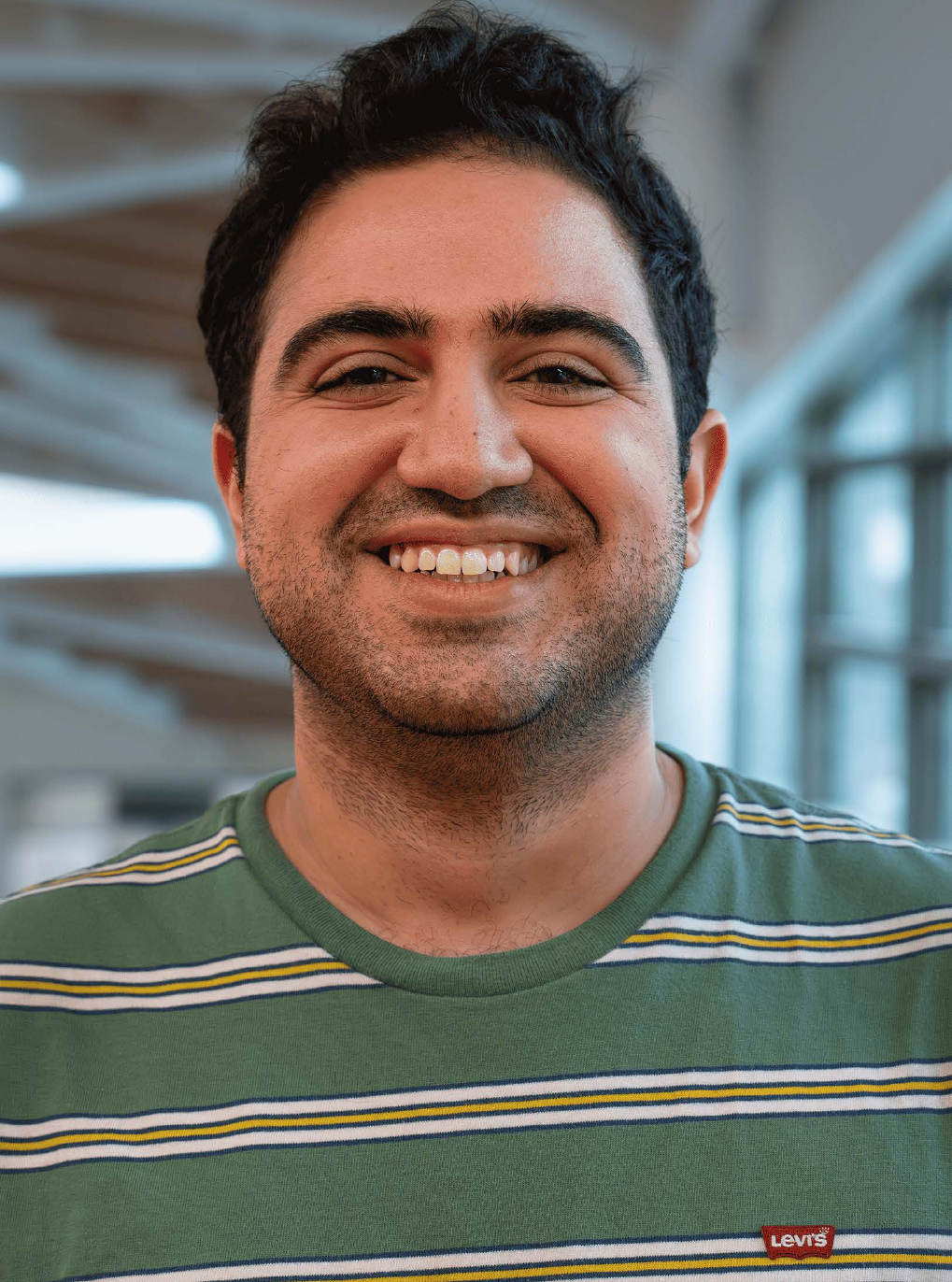}}]{Behrad Samari}~(Student Member, IEEE) received his B.Sc. and M.Sc. degrees in electrical engineering, control major, from K. N. Toosi University of Technology, Tehran, Iran, and University of Tehran (UT), Tehran, Iran, in 2019 and 2022, respectively. He is currently pursuing his PhD in the School of Computing at Newcastle University, U.K. He is the Best Repeatability Prize Finalist at the 8$^{\text{th}}$ IFAC Conference on Analysis and Design of Hybrid Systems (ADHS), 2024. His research interests include (nonlinear) control and system theory, data-driven approaches, and formal methods.
\end{IEEEbiography}

\begin{IEEEbiography}[{\includegraphics[width=1in,height=1.25in,clip,keepaspectratio]{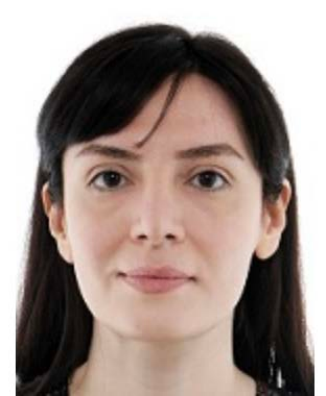}}]{Amy Nejati}~(M'18--SM'25) is an Assistant Professor in the School of Computing at Newcastle University in the United Kingdom. Prior to this, she was a Postdoctoral Associate at the Max Planck Institute for Software Systems in Germany from July 2023 to May 2024. She also served as a Senior Researcher in the Computer Science Department at the Ludwig Maximilian University of Munich (LMU) from November 2022 to June 2023. She received the PhD in Electrical Engineering from the Technical University of Munich (TUM) in 2023. She has received the B.Sc. and M.Sc. degrees both in Electrical Engineering. Her line of research mainly focuses on developing efficient (data-driven) techniques to design and control highly-reliable autonomous systems while providing mathematical guarantees. She was selected as one of the CPS Rising Stars 2024.\end{IEEEbiography}

\begin{IEEEbiography}[{\includegraphics[width=1in,height=1.25in,clip,keepaspectratio]{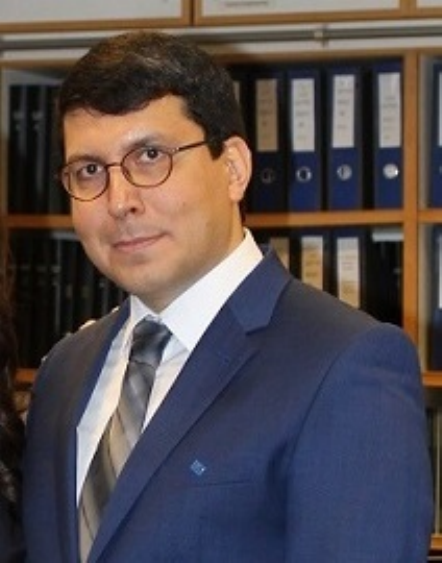}}]{Abolfazl Lavaei}~(M'17--SM'22) is an Assistant Professor in the School of Computing at Newcastle University, United Kingdom. Between January 2021 and July 2022, he was a Postdoctoral Associate in the Institute for Dynamic Systems and Control at ETH Zurich, Switzerland. He was also a Postdoctoral Researcher in the Department of Computer Science at LMU Munich, Germany, between November 2019 and January 2021. He received the Ph.D. degree in Electrical Engineering from the Technical University of Munich (TUM), Germany, in 2019. He obtained the M.Sc. degree in Aerospace Engineering with specialization in Flight Dynamics and Control from the University of Tehran (UT), Iran, in 2014. He is the recipient of several international awards in the acknowledgment of his work including  Best Repeatability Prize (Finalist) at the ACM HSCC 2025, IFAC ADHS 2024, and IFAC ADHS 2021, HSCC Best Demo/Poster Awards 2022 and 2020, IFAC Young Author Award Finalist 2019, and Best Graduate Student Award 2014 at University of Tehran with the full GPA (20/20). His line of research primarily focuses on the intersection of Formal Methods in Computer Science, Control Theory, and Data Science.
\end{IEEEbiography}

\end{document}